\documentclass[11pt,reqno]{article}

\usepackage[leqno]{amsmath}
\usepackage{amsthm}
\usepackage{amssymb}
\usepackage{times}
\usepackage{subcaption}
\usepackage{verbatim}
\usepackage{enumerate}
\usepackage[margin=3.2cm]{geometry}
\usepackage[pdftex]{graphicx}
\usepackage{xspace}
\allowdisplaybreaks[1] 

\newcommand\myeq{\mathrel{\overset{\makebox[0pt]{\mbox{\normalfont\tiny\sffamily def}}}{=}}}

\newtheorem{theorem}{Theorem}[section]

\newtheorem{lemma}[theorem]{Lemma}
\newtheorem{corollary}[theorem]{Corollary}

\newtheorem{definition}{{\sc Definition}\rm}[section]

\newcommand{\ED}{\textsc{ED}\xspace}
\newcommand{\E}{\mathbb{E}}
\newcommand{\cl}[1]{\ensuremath {\sf #1}}

\begin{document}

\title {Element Distinctness, Frequency Moments, and Sliding Windows\vspace*{-2ex}}
\author{Paul Beame\thanks{Research supported by NSF grants CCF-1217099 and CCF-0916400}\\
  {\small Computer Science and Engineering}\\
  {\small University of Washington}\\
  {\small Seattle, WA 98195-2350}\\
  {\small\tt beame@cs.washington.edu}
  \and
  Rapha\"el Clifford\thanks{Research supported by the EPSRC.  This work was done while the author was visiting the University of Washington.} \\
  {\small Department of Computer Science}\\
  {\small University of Bristol}\\
  {\small Bristol BS8 1UB, United Kingdom}\\
  {\small\tt clifford@cs.bris.ac.uk }\\
  \and
  Widad Machmouchi\thanks{Research supported by NSF grant CCF-1217099}\\
  {\small Computer Science and Engineering}\\
  {\small University of Washington}\\
  {\small Seattle, WA 98195-2350}\\
  {\small\tt widad@cs.washington.edu}
  \vspace*{-2ex}}

\date{}

\maketitle
\thispagestyle{empty}
\begin{abstract}
We derive new time-space tradeoff lower bounds and algorithms for exactly
computing statistics of input data, including frequency moments,
element distinctness, and order statistics, that are simple to calculate for
sorted data. 
In particular, we develop a randomized algorithm for the element distinctness
problem whose time $T$ and space $S$ satisfy $T\in \tilde O(n^{3/2}/S^{1/2})$,
smaller than previous lower bounds for comparison-based algorithms, 
showing that element distinctness is strictly easier than sorting for
randomized branching programs.
This algorithm is based on a new time- and space-efficient algorithm for
finding all collisions of a function $f$ from a finite set to itself that are
reachable by iterating $f$ from a given set of starting points.

We further show that our element distinctness algorithm can be extended at only
a polylogarithmic factor cost to solve the
element distinctness problem over sliding windows~\cite{dgim:windows}, where
the task is to take an input of length $2n-1$ and produce an output for
each window of length $n$, giving $n$ outputs in total.

In contrast, we show a time-space tradeoff lower
bound of $T\in \Omega(n^2/S)$ for randomized multi-way branching programs, and
hence standard RAM and word-RAM models, to compute the number
of distinct elements, $F_0$, over sliding windows. 
The same lower bound holds for computing the low-order bit of $F_0$ and
computing any frequency moment $F_k$ for $k\ne 1$.
This shows that frequency moments $F_k\ne 1$ and even the decision
problem $F_0\bmod 2$ are strictly harder than element distinctness.
We provide even stronger separations on average for inputs from $[n]$.

We complement this lower bound with a $T\in \tilde O(n^2/S)$
comparison-based deterministic RAM algorithm for exactly computing $F_k$ over
sliding windows, nearly matching both our general lower bound for
the sliding-window version and the comparison-based lower bounds for a single
instance of the problem.   We further exhibit a quantum algorithm for
$F_0$ over sliding windows with $T\in \tilde O(n^{3/2}/S^{1/2})$.
Finally, we consider the computations of order statistics over sliding windows.
\end{abstract}

\newpage
\setcounter{page}{1}
\section{Introduction}

Problems related to computing elementary statistics of input data have wide
applicability and utility.  
Despite their usefulness, there are surprising gaps in our knowledge about the
best ways to solve these simple problems, particularly in the context of
limited storage space.
Many of these elementary statistics can be easily calculated if the input data
is already sorted but sorting the data in space $S$,
requires time $T\in \Omega(n^2/S)$~\cite{bc:sorting,bea:sorting}, 
a bound matched by the best comparison
algorithms~\cite{pr:comparison-sorting} for all $S\in O(n/\log n)$.  
It has not been clear whether exactly computing elementary statistical
properties, such as frequency moments (e.g. $F_0$, the number of distinct
elements in the input) or element distinctness ($ED$, whether or not $F_0$
equals the input size) are as difficult as sorting when
storage is limited.

The main approach to proving time-space tradeoff lower bounds for problems
in $\cl{P}$ has been to analyze their complexity on (multi-way) \emph{branching
programs}.
(As is usual, the input is
assumed to be stored in read-only memory and the output in write-only memory
and neither is counted towards the space used by any algorithm.
The multi-way branching program model simulates both Turing machines and 
standard RAM models that
are unit-cost with respect to time and log-cost with respect to space.)
An important method for this analysis 
was introduced by Borodin and Cook for sorting~\cite{bc:sorting} 
and has since been extended and generalized to randomized computation of
a number of other important multi-output problems (e.g.,
\cite{yes84,abr87,abr:tradeoff,bea:sorting,mnt:hashing-journal,sw:multiplyRAM}).
Unfortunately, the techniques of~\cite{bc:sorting} yield only trivial
bounds for problems with single outputs such as $F_0$ or $ED$.

Element distinctness has been a particular focus of lower bound analysis.
The first time-space tradeoff lower bounds for the problem apply to
structured algorithms.
Borodin et al.~\cite{bfmuw87} gave a time-space tradeoff lower bound for
computing $ED$ on \emph{comparison}
branching programs of $T\in\Omega(n^{3/2}/S^{1/2})$ and, since
$S\geq \log_2{n}$, $T\in\Omega(n^{3/2}\sqrt{\log n}/S)$.
Yao~\cite{yao88} improved this to a near-optimal
$T\in \Omega(n^{2-\epsilon(n)}/S)$, where $\epsilon(n)= 5/(\ln n)^{1/2}$.
Since these lower bounds apply to the average case for randomly ordered inputs,
by Yao's lemma, they also apply to randomized comparison branching programs.
These bounds also trivially apply to all frequency moments 
since, for $k\ne 1$, $ED(x)=n$ iff $F_k(x)=n$.
This near-quadratic
lower bound seemed to suggest that the complexity of $ED$ and $F_k$
should closely track that of sorting.

\begin{sloppypar}
For multi-way branching programs, Ajtai~\cite{ajtai:nonlinear-journal}
showed that any linear time algorithm for $ED$ 
must consume linear space.   Moreover, when $S$ is $n^{o(1)}$, Beame et
al.~\cite{bssv:randomts-journal}
showed a $T \in \Omega(n\sqrt{\log (n/S)/\log\log(n/S)})$
lower bound for computing $ED$.
This is a long way from the comparison branching program lower bound and
there has not been much prospect for closing the gap
since the largest lower bound known for multi-way branching programs
computing \emph{any} single-output problem in $\cl{P}$ is only 
$T\in \Omega(n\log ((n\log n)/S))$.
\end{sloppypar}

We show that this gap between sorting and element distinctness cannot be closed.
More precisely, we give a randomized multi-way branching program algorithm that
for any space bound $S\in [c\log n, n]$ computes $ED$ in time
$T\in \tilde O(n^{3/2}/S^{1/2})$\footnote{As is usual, we use $\tilde O$ to
suppress polylogarithmic factors in $n$.},
significantly beating the lower bound that
applies to comparison-based algorithms.
Our algorithm for $ED$ is based on an extension of Floyd's cycle-finding
algorithm~\cite{knu2} (more precisely, its variant, Pollard's rho
algorithm~\cite{pollard75}).
Pollard's rho algorithm finds the unique collision reachable by iterating
a function $f:[n]\rightarrow [n]$ from a single starting location in time
proportional
to the size of the reachable set, using only a constant number of pointers.
Variants of this algorithm have been used in cryptographic
applications to find collisions in functions that supposedly behave like
random functions~\cite{brent80, ssy82,nivasch04}.

More precisely, our new $ED$ algorithm is based on a new deterministic
extension of Floyd's algorithm to find \emph{all} collisions of a function
$f:[n]\rightarrow [n]$ reachable
by iterating $f$ from any one of a set of $k$ starting locations,
using only $O(k)$ pointers and using time roughly proportional to the size of
the reachable set.
Motivated by cryptographic applications, \cite{ow99} previously considered this
problem for the special case of random functions and suggested a method using
`distinguished points',
though the only analysis they gave was heuristic and incomplete. 
Our algorithm, developed independently, uses a different method,
applies to arbitrary functions, and has a fully rigorous analysis.

Our algorithm for $ED$ does not obviously apply to the computation of
frequency moments, such as $F_0$,
and so it is interesting to ask whether or not frequency moment computation
is harder than that of $ED$ and may be closer in complexity to sorting. 
Given the general difficulty of obtaining strong lower bounds for
single-output functions, we consider the relative complexity of computing
many copies of each of the functions at once and apply techniques for
multi-output functions to make the comparison.   
Since we want to retain a similar input size to that of our original problems,
we need to evaluate them on overlapping inputs.  

Evaluating the same function on overlapping inputs occurs as a natural
problem in time series analysis when it is useful to know the value of a
function on many
different intervals or {\em windows} within a sequence of values or updates,
each representing the recent history of the data at a given instant.
In the case that an answer
for every new element of the sequence is required, such computations have
been termed {\em sliding-window} computations for the associated
functions~\cite{dgim:windows}.
In particular, we consider inputs of length $2n-1$ where
the sliding-window task is to compute the function
for each window of length $n$, giving $n$ outputs in total.
We write $F^{\boxplus n}$ to denote this sliding-window version of a function
$F$.   

Many natural functions have been studied for sliding windows
including entropy, finding frequent symbols, frequency moments and order
statistics, which can be computed approximately in small
space using randomization even in one-pass data stream
algorithms~\cite{dgim:windows,BabcockBDMW02,am04,lt06, lt06b, ccm07,boz12}.
Approximation is required since exactly computing these values in this online
model can easily be shown to require large space.
The interested reader may find a more comprehensive list of sliding-windows
results by following the references in~\cite{boz12}.

%

We show that computing $ED$ over $n$ sliding windows only incurs a
polylogarithmic overhead in time and space versus computing a single copy of
$ED$.
In particular, we can extend our randomized multi-way branching program
algorithm for $ED$ to yield an algorithm for $ED^{\boxplus n}$ that for 
space $S\in [c\log n,n]$ runs in time $T\in \tilde O(n^{3/2}/S^{1/2})$.

In contrast, we prove strong time-space lower bounds for computing the
sliding-window version
of any frequency moment $F_k$ for $k\ne 1$.  In particular, the time $T$ and
space $S$ to compute $F_k^{\boxplus n}$ must satisfy
$T\in \Omega(n^2/S)$ and $S\ge \log n$.
($F_1$ is simply the size of the input, so computing its value is always
trivial.)
The bounds are proved directly for randomized multi-way branching programs
which imply
lower bounds for the standard RAM and word-RAM models, as well as for the data
stream models discussed above.
Moreover, we show that the same lower bound holds for computing just
the parity of the number of distinct elements, $F_0\bmod 2$, in each window.
This formally proves a separation between the complexity of sliding-window
$F_0\bmod 2$ and sliding-window $ED$.
These results suggest that in the continuing search for strong complexity
lower bounds, $F_0\bmod 2$ may be a better choice as a difficult decision
problem than $ED$.

Our lower bounds for frequency moment computation hold for randomized
algorithms even with small success probability $2^{-O(S)}$ and for the average
time and space used by deterministic algorithms on inputs in which the values
are independently and uniformly chosen from $[n]$.  (For comparison with the
latter average case results, it is not hard to show that over the same input
distribution
$ED$ can be solved with $\overline T \in \tilde O(n/\overline S )$ and
our reduction shows that this can be extended to
$\overline T \in \tilde O(n/\overline S)$ bound for $ED^{\boxplus n}$
on this input distribution.)

We complement our lower bound with a comparison-based RAM algorithm for
any $F_k^{\boxplus n}$ that has $T\in \tilde O(n^2/S)$, showing that
this is nearly an asymptotically tight bound, since it provides a general RAM
algorithm that runs in the same time complexity.
Since our algorithm for computing $F_k^{\boxplus n}$ is comparison-based,
the comparison lower bound for $F_k$ implied by~\cite{yao88} is not far from
matching our algorithm even for a single instance of $F_k$.
We also provide a quantum algorithm for $F_0^{\boxplus n}$ with
$T\in \tilde O(n^{3/2}/S^{1/2})$.

It is interesting to understand how the complexity of computing a function $F$
can be related to that of computing $F^{\boxplus n}$.
To this end, we consider problems of computing the $t^{th}$ order statistic in
each window.  For these problems we see the full range of relationships between
the complexities of the original and sliding-window versions of the problems.
In the case of $t=n$ (maximum) or $t=1$ (minimum) we show that
computing these properties over sliding windows can be solved by a comparison
based algorithm in $O(n\log n)$
time and only $O(\log n)$ bits of space so there is very little growth in
complexity.
In contrast, we show that a
$T \in \Omega(n^2/S)$ lower bound holds
when $t  = \alpha n$ for any fixed $0 < \alpha < 1$.
Even for algorithms that only use comparisons, the expected time for
errorless randomized algorithms to find the median in a single window is 
$\overline T\in \Theta(n\log\log_S n)$~\cite{chan:selection-journal}.
Hence, these problems have a dramatic
increase in complexity over sliding windows.\\[-5ex]

\paragraph{Related work}
While sliding-windows versions of problems have been considered in the context
of online and approximate computation, there is little research that has
explicitly considered any such problems in the case of exact offline
computation.
One instance where a sliding-windows problem has been considered is
a lower bound for generalized string
matching due to Abrahamson~\cite{abr87}.
This lower bound implies that for any fixed string $y\in [n]^n$
with $n$ distinct values, $H_y^{\boxplus n}$ requires
$T\cdot S\in \Omega(n^2/\log n)$ where decision problem $H_y(x)$ is 1 if and
only if the Hamming distance between $x$ and $y$ is $n$.
This bound is an $\Omega(\log n)$
factor smaller than our lower bound for sliding-window $F_0\bmod 2$.

\paragraph{Frequency Moments, Element Distinctness, and Order Statistics}
Let $a =  a_1a_2\ldots a_n\in D^n$ for some finite set $D$. 
We define the \textit{$k^{th}$ frequency
moment} of $a$, $F_k(a)$, as $F_k(a)= \sum_{i\in D} f_i^k$, where $f_i$ is the
frequency (number of occurrences) of symbol $i$ in the string $a$ and $D$ is
the set of symbols that occur in $a$.
Therefore, $F_0(a)$ is the number of distinct
symbols in $a$ and $F_1(a) = |a|$ for every string $a$.
The \textit{element distinctness} problem is a decision problem
defined as: $ED(a) =1  \mbox{ if } F_0(a) = |a| \mbox{ and } 0
\mbox{ otherwise}.$
We write $ED_n$ for the $ED$ function restricted to inputs $a$ with $|a|=n$.
The \textit{$t^{th}$ order statistic} of $a$,
$O_t$, is the $t^{th}$ smallest symbol in $a$.
Therefore $O_n$ is
the maximum of the symbols of $a$ and $O_{\lceil \frac{n}{2}\rceil}$
is the median.\\[-5ex]

\paragraph{Branching programs}
Let $D$ and $R$ be finite sets and $n$ and $m$ be two positive integers.
A \textit{$D$-way branching program} is a
connected directed acyclic graph with special nodes: the
\textit{source node} and possibly many \textit{sink nodes}, a set of
$n$ inputs and $m$ outputs.
Each non-sink node is labeled with an input index and every edge is
labeled with a symbol from $D$, which corresponds to the value of
the input indexed at the originating node.
In order not to count the space required for outputs, as is standard for the
multi-output problems~\cite{bc:sorting},
we assume that each edge can be labeled by some set of output assignments.
For a directed path $\pi$ in a branching program, we call the set of indices of
symbols queried by $\pi$ the \emph{queries} of $\pi$, denoted by $Q_{\pi}$;
we denote the \emph{answers} to those queries by $A_{\pi}:Q_\pi \rightarrow D$
and the outputs produced along $\pi$ as a \emph{partial}
function $Z_\pi:[m]\rightarrow R$.

A branching program computes a function $f:D^n\rightarrow R^m$ by
starting at the source and then proceeding along the nodes of the
graph by querying the inputs associated with each node and following
the corresponding edges.
In the special case that there is precisely one output, without loss of
generality, any edge with this output may instead be assumed to be unlabeled
and lead to a unique sink node associated with its output value.

A branching program $B$ computes a function $f$
if for every $x \in D^n$, the output of $B$ on $x$, denoted $B(x)$,
is equal to $f(x)$.
A \textit{computation} of $B$ on $x$ is a
directed path, denoted $\pi_B(x)$, from the source to a sink in $B$ whose
queries to the input are consistent with $x$.
The time $T$ of a branching program is the length of the
longest path from the source to a sink and the space $S$ is the
logarithm base 2 of the number of the nodes in the branching
program. Therefore, $S \geq \log T$ where we write $\log x$ to denote
$\log_2 x$.

A \textit{randomized} branching program
$\mathcal{B}$ is a probability distribution over deterministic
branching programs with the same input set. $\mathcal{B}$
computes a function $f$ with error at most $\eta$ if for every
input $x \in D^n$, $\Pr_{B\sim\mathcal{B}}[B(x) = f(x)] \geq
1-\eta$. The time (resp. space) of a randomized branching program
is the maximum time (resp. space) of a deterministic branching program
in the support of the distribution.  

While our lower bounds apply to randomized branching programs,
which allow the strongest non-explicit randomness, our randomized algorithms
for element distinctness will only require a weaker notion, {\em input randomness},
in which the random string $r$ is given as an explicit input to a RAM
algorithm.  For space-bounded computation, it would be preferable 
to only require the random bits to be available online as the algorithm
proceeds.

As is usual in analyzing randomized computation via Yao's lemma, we also will
consider complexity under distributions $\mu$ on the input space $D^n$.
A branching program $B$ \textit{computes $f$ under
$\mu$ with error at most $\eta$} iff $B(x)= f(x)$ for all but an
$\eta$-measure of $x \in D^n$ under distribution $\mu$.

A branching program is \textit{leveled} if  the nodes are divided
into an ordered collection of sets each called a \textit{ level}
where edges are between consecutive levels only.
Any branching program can be leveled by increasing the space
$S$ by an additive $\log T$. Since $S \geq \log T$, in
the following we assume without loss of generality that our branching programs
are leveled.

A \textit{comparison branching program} is similar to a $D$-way
branching program except that each non-sink
node is labeled with a pair of indices
$i,j\in [n]$ with $i<j$ and has precisely three out-edges with labels from
$\{<,=,>\}$,
corresponding to the relative position of inputs $x_i$ and $x_j$ in the
total order of the inputs.

\section{Element Distinctness and Small-Space Collision-Finding }\label{sec:ed}

\subsection{Efficient small-space collision-finding with many sources}

Our approach for solving the element distinctness problem has at its heart a
novel extension of Floyd's small space ``tortoise and hare'' cycle-finding
algorithm~\cite{knu2}.
Given a start vertex $v$ in a finite graph $G=(V,E)$ of outdegree 1,
Floyd's algorithm finds the unique cycle in $G$
that is reachable from $v$.
The out-degree 1 edge relation $E$ can be viewed as a
set of pairs $(u,f(u))$ for a function $f:V\rightarrow V$. 
Floyd's algorithm, more precisely, stores only two values from $V$ and
finds the smallest $s$ and $\ell>0$ and vertex $w$ such
such that $f^s(v)=f^{s+\ell}(v)=w$ using only $O(s+\ell)$ evaluations of $f$.

We say that vertices $u\ne u'\in V$ are {\em colliding} iff
$f(u)=f(u')$ and call $v=f(u)=f(u')$ a collision.
Floyd's algorithm for cycle-finding can also be useful for finding collisions since
in many instances the starting vertex $v$ is not on a cycle and thus $s>0$.
In this case $i=f^{s-1}(v)\ne j=f^{s+\ell-1}(v)$ satisfy $f(i)=f(j)=w$, which is
a collision in $f$, and the iterates of $f$ produce a $\rho$ shape (see
Figure~\ref{fig:multiFloyda}).  
These colliding points may be found with minimal cost
by also storing the previous values of each of the two pointers as Floyd's
algorithm proceeds.
The $\rho$ shape inspired the name of
Pollard's rho algorithm for factoring~\cite{pollard75} and solving discrete logarithm problems~\cite{pollard78}. This in turn is also the commonly associated name for the application of Floyd's cycle finding algorithm to the general problem of collision-finding. 

Within the cryptographic community there is an extensive body of work building
on these early advances of Pollard.
There is also considerable research on cycle detection algorithms that use
larger space than Floyd's algorithm  and improve the constant factors in the
number of edges that must be traversed (function evaluations) 
to find the cycle (see for example~\cite{brent80, ssy82,nivasch04} and
references therein).
In these applications, the goal is to find a single collision reachable from a given starting point as efficiently as possible.  

The closest existing work to our problem tackles the problem of speeding up
collision detection in random functions using parallelization.
In~\cite{ow99}, van Oorschot and Wiener
give a deterministic parallel algorithm for finding all collisions
of a random function using $k$ processors, along with a heuristic analysis of
its performance.  
Their algorithm keeps a record of visits to predetermined  vertices
(so-called `distinguished points')  which allow the separate processes 
to determine quickly if they are on a path that has been explored before. 
Their heuristic argument suggests  a bound of $O(n^{3/2}/k^{1/2})$
function iterations, though it is unclear whether this can ever be made
rigorous. 
Our method is very different in detail and developed independently: 
we provide a fully rigorous analysis that roughly matches the heuristic bound
of~\cite{ow99} by using a new deterministic algorithm for collision-finding
on random hash functions applied to worst-case inputs for element distinctness.

As part of our strategy for efficiently solving element distinctness,
we examine the time and space complexity of finding {\em all} colliding
vertices, along with their predecessors, in the subgraph reachable from a
possibly large set of $k$ starting vertices,
not just from a single start vertex.
We will show how to do this using storage equivalent to only $O(k)$ elements
of $V$ and time roughly proportional to the size of the subgraph
reachable from this set of starting vertices.
Note that the obvious approach of running $k$ independent copies of Floyd's
algorithm in parallel from each of the start vertices does not solve this
problem since it may miss collisions between different parallel branches (see
Figure~\ref{fig:multiFloydb}), and
it may also traverse large regions of the subgraph many times.

For $v\in V$, define $f^*(v)=\{f^i(v)\mid i\ge 0\}$ to be the set of
vertices reachable from $v$ and $f^*(U)=\bigcup_{v\in U} f^*(v)$ for
$U\subseteq V$. 

\begin{figure}[t]
\centering
\begin{tabular}{@{}c|c@{}}
   \subcaptionbox{\label{fig:multiFloyda}}{\includegraphics[width = 0.45\columnwidth, height=0.15\columnwidth]{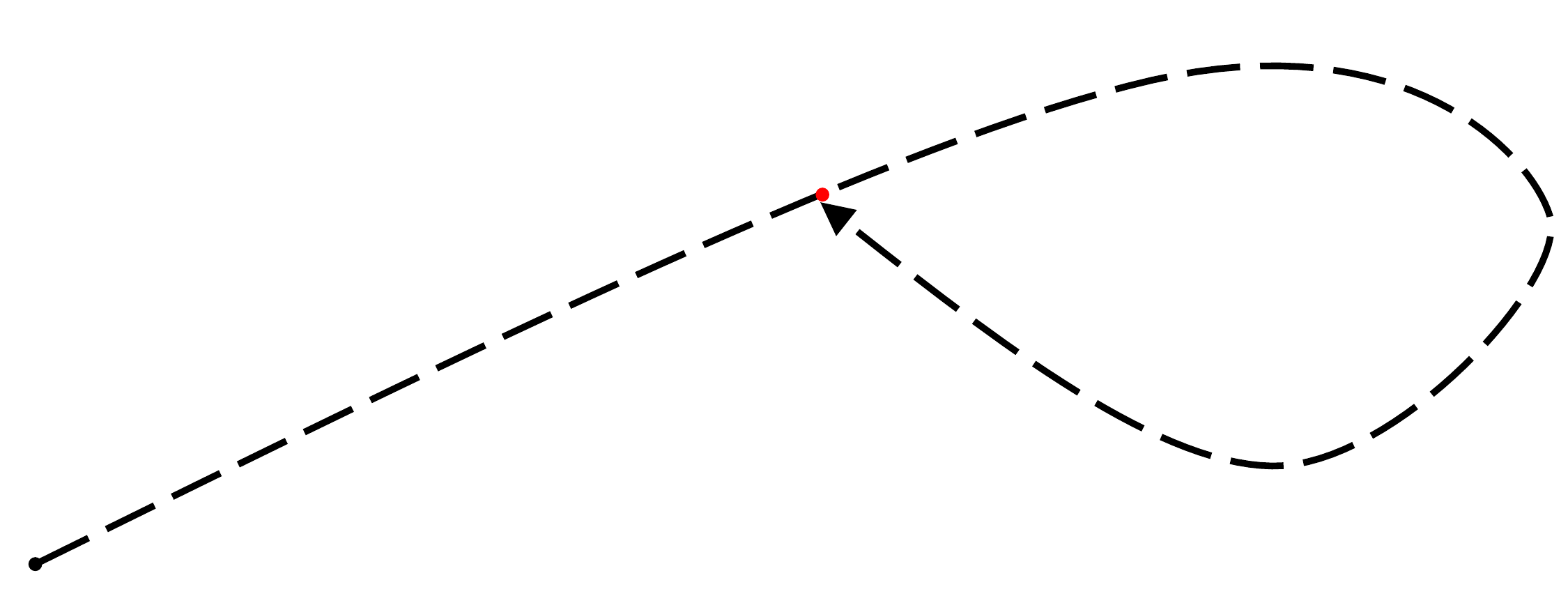}}&
   \subcaptionbox{\label{fig:multiFloydc}}{\includegraphics[width = 0.45\columnwidth, height=0.15\columnwidth]{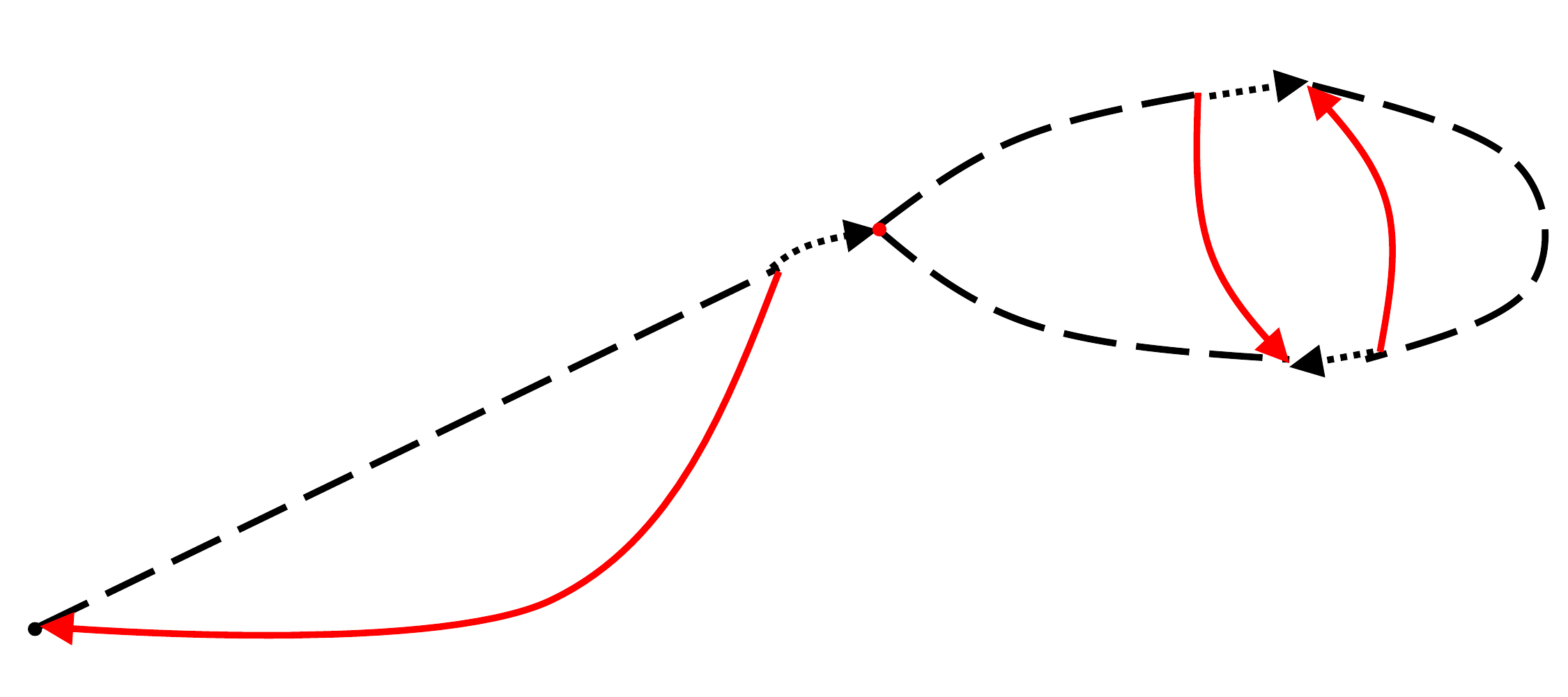}}\\
   \hline
  & \\[-1.88ex]
   \subcaptionbox{\label{fig:multiFloydb}}{\includegraphics[width = 0.45\columnwidth, height=0.25\columnwidth]{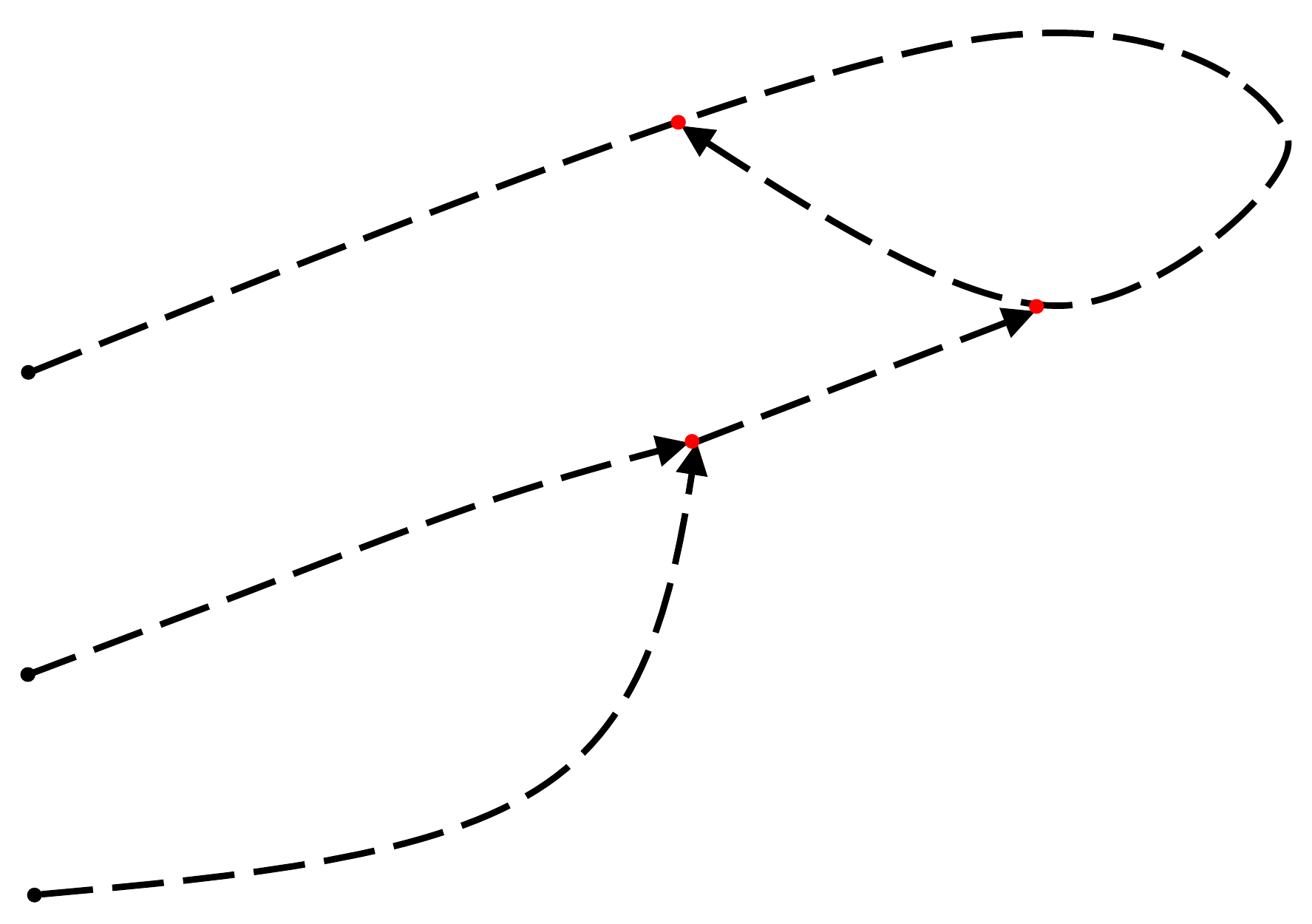}}&
   \subcaptionbox{\label{fig:multiFloydd}}{\includegraphics[width = 0.45\columnwidth, height=0.25\columnwidth]{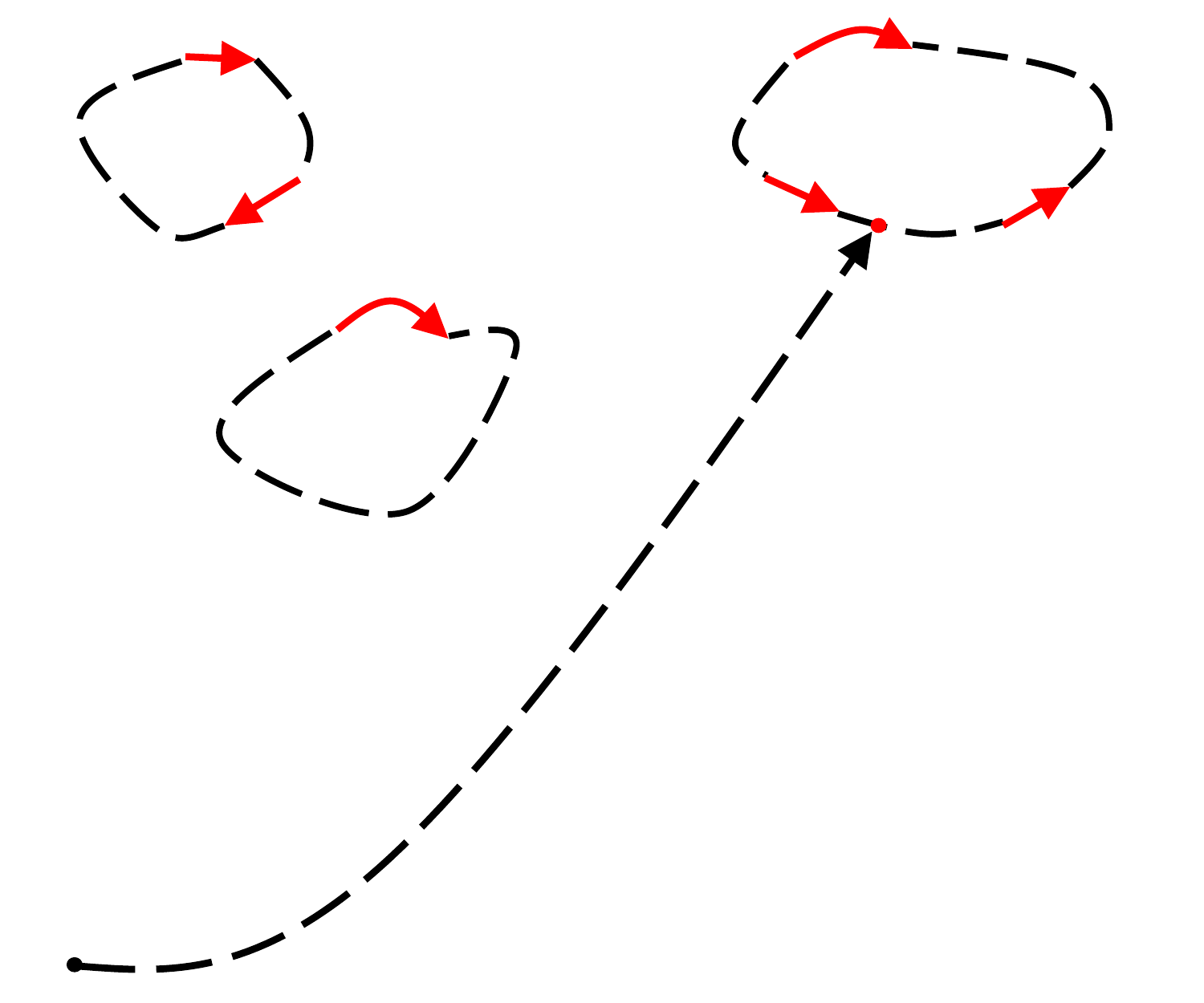}}
   \end{tabular}
\caption{\small Collision-finding with multiple sources}
\label{fig:multiFloyd}
\end{figure}

\begin{theorem}
\label{graph-multiFloyd}
There is an $O(k\log n)$ space deterministic algorithm {\sc Collide}$_k$ that,
given $f:V\rightarrow V$ for a finite set $V$ and
$K=\{v_1,\ldots, v_k\}\subseteq V$, 
finds all pairs $\left(v,\left\{u \in f^*(K) \middle| f(u) = v\right\}\right)$  and runs in time
$O(| f^*(K)|\log k \min\{k,\log n\})$.
\end{theorem}

\begin{proof}
We first describe the algorithm {\sc Collide}$_k$:
In addition to the original graph and the collisions that it finds, this
algorithm will maintain a {\em redirection list} $R\subset V$ of size $O(k)$ 
vertices that it will provisionally redirect to map to a different location.
For each vertex in $R$ it will store the name of the new vertex to which it is
directed.   We maintain a separate list $L$ of all vertices from which an edge
of $G$ has been redirected away and the original vertices that point to them.

\noindent
{\sc Collide}$_k$:\\
Set $R=\emptyset$.\\
For $j=1,\ldots,k$ do:
\begin{enumerate}
\item Execute Floyd's algorithm starting with vertex $v_j$ on the graph $G$
using the redirected out-edges for nodes from the redirection list $R$ instead
of $f$.
\item If the cycle found does not include $v_j$, there must be a collision.
\begin{enumerate}
\item If this collision is in the graph $G$, report the collision $v$ as well as the colliding vertices $u$ and $u'$, where $u'$ is the predecessor of $v$ on the cycle and $u$ is the
predecessor of $v$ on the path from $v_j$ to $v$.

\item Add $u$ to the redirection list $R$ and redirect it to vertex $v_j$.
\end{enumerate}
\item Traverse the cycle again to find its length and choose two
vertices $w$ and $w'$ on the cycle that are within 1 of half this length apart.
Add $w$ and $w'$ to the redirection list, redirecting $w$ to $f(w')$ and
$w'$ to $f(w)$.   This will split the cycle into two parts, each of roughly
half its original length.
\end{enumerate}
The redirections for a single iteration of the algorithm are shown in
Figure~\ref{fig:multiFloydc}.
The general situation for later iterations in the algorithm is shown in
Figure~\ref{fig:multiFloydd}.

Observe that in each iteration of the loop there is at most one vertex $v$
where collisions can occur and at most 3 vertices are added to the
redirection list.  Moreover, after each iteration,
the set of vertices reachable from vertices $v_1,\ldots, v_j$ appear in
a collection of disjoint cycles of the redirected graph.
Each iteration of the loop traverses at most one cycle and
every cycle is roughly halved each time it is traversed.

In order to store the redirection list $R$, we use a dynamic dictionary
data structure of $O(k\log n)$ bits that supports insert and search in
$O(\log k)$ time per access or insertion.
We can achieve this using balanced binary search trees and we can improve
the bound to $O{(\sqrt{\log k/\log\log k})}$ using exponential
trees~\cite{at:exponentialtrees}.
Before following an edge (i.e., evaluating $f$), the algorithm will first
check list $R$ to see if it has been redirected.  Hence each edge traversed
costs $O(\log k)$ time (or $O{(\sqrt{\log k/\log\log k})}$ using
exponential trees).  Since time is measured relative to the size of the
reachable set of vertices,
the only other extra cost is that of re-traversing previously discovered edges.
Since all vertices are maintained in cycles and each traversal of a cycle
roughly halves its length, each edge found can be traversed at most
$O(\min\{k,\log n\})$ times.  
\end{proof}

As we have noted in the proof, with exponential trees the running time
of the algorithm can be reduced to
$O(| f^*(K)|\sqrt{\log k/\log\log k} \min\{k,\log n\})$, though the algorithm 
becomes significantly more complicated in this case.

\subsection{A randomized $T^2\cdot S\in \tilde O(n^3)$ element distinctness algorithm}
\label{sec:fastED}

We will use collision-finding for our element distinctness algorithm. 
In this case the vertex set $V$ will be the set of indices $[n]$, and the
function $f$ will be given by $f_{x,h}(i)=h(x_i)$ where $h$ is
a (random) hash function that maps $[m]$ to $[n]$.    

Observe that if we find $i\ne j$ such that $f_{x,h}(i)=f_{x,h}(j)$ then either
\begin{itemize}
\item $x_i=x_j$ and hence $\ED(x)=0$, or 
\item we have found a collision in $h$: $x_i\ne x_j$ but $h(x_i)=h(x_j)$;
\end{itemize}
We call $x_i$ and $x_j$ ``pseudo-duplicates" in this latter case.

Given a parameter $k$, on input $x$ our randomized algorithm will repeatedly
choose a random hash
function $h$ and a random set $K$ of roughly $k$ starting points and then
call the {\sc Collide}$_k$ algorithm given in
Theorem~\ref{graph-multiFloyd} on $K$ using the function $f=f_{x,h}$
and check the collisions found to determine whether or not there is a
duplicate among the elements of $x$ indexed by $f_{x,h}^*(K)$.
The space bound $S$ of this algorithm will be $O(k\log n)$. 

The running time of {\sc Collide}$_k$ depends on 
$|f_{x,h}^*(K)|$, which in turn is governed by the random choices of $h$
and $K$ and may be large.   
Since $f_{x,h}^*(K)$ is also random, we also need to argue that if $ED(x)=0$,
then there is a reasonable probability that a duplicate in $x$ will be
found among the indices in $f_{x,h}^*(K)$.
The following two lemmas analyze these issues.

\begin{lemma}
\label{closure}
Let $x\in [m]^n$.   For $h:[m]\rightarrow [n]$ chosen uniformly at random
and for $K\subseteq [n]$ selected by uniformly and independently
choosing $2\le k\le n/32$ elements of $[n]$ with replacement,
$\Pr[|f_{x,h}^*(K)|\le 2\sqrt{kn}]\ge 8/9$.
\end{lemma}

\begin{lemma}
\label{found-duplicate}
Let $x\in [m]^n$ be such that $ED(x)=0$.  
Then for $h:[m]\rightarrow [n]$ chosen uniformly at random
and for $K\subseteq [n]$ selected by uniformly and independently
choosing $2\le k\le n/32$ elements of $[n]$ with replacement,
$\Pr[ |f_{x,h}^*(K)|\le 2\sqrt{kn} \mbox{ and  
$\exists i\ne j\in f_{x,h}^*(K)$ such that $x_i=x_j$}]\ge k/(18n)$.
\end{lemma}

Before we prove these lemmas we show how, together with the properties of 
{\sc Collide}$_k$, they yield the following theorem.

\begin{theorem}
\label{EDtradeoff}
For any $\epsilon>0$, and any $S$ with $c\log n\le S\le n/32$ for
some constant $c>0$, there is a randomized RAM algorithm with input randomness
computing  $ED_n$ with 1-sided error (false positives) at most
$\epsilon$,
that uses space $S$ and time $T\in O(\frac{n^{3/2}}{S^{1/2}}\log^{5/2} n \log (1/\epsilon))$.
Further, when $S\in O(\log n)$, we have
$T\in O(n^{3/2} \log (1/\epsilon))$.
\end{theorem}

\begin{proof}
Choose $k\ge 2$ such that the space usage of {\sc Collide}$_k$ on $[n]$ is at
most $S/2$.    Therefore $k\in \Omega(S/\log n)$.  The algorithm is as follows:

\begin{verse}
On input $x$, run $(18n/k) \log (1/\epsilon)$ 
independent runs of {\sc Collide}$_k$ on different $f_{x,h}$, each with
independent random choices 
of hash functions $h$ and independent choices, $K$, of $k$ starting indices,
and each with a
run-time cut-off bounding the number of explored vertices of $f_{x,h}^*(K)$ at
$t^*=2\sqrt{kn}$.  \\
On each run, check if any of the collisions found is a duplicate in $x$,
in which case output $ED(x)=0$ and halt.  If none are found in any round then
output $ED(x)=1$.
\end{verse}

The algorithm will never incorrectly report a duplicate in a distinct $x$
and by Lemma~\ref{found-duplicate}, each run has a probability of at least
$k/(18n)$ of finding a duplicate in an input $x$ such that $ED(x)=0$ so the 
probability of failing to find a duplicate in $(18n/k) \log (1/\epsilon)$ rounds
is at most $\epsilon$.

\begin{sloppypar}
Using Theorem~\ref{graph-multiFloyd} each run of our bounded version of
{\sc Collide}$_k$, requires runtime $O(\sqrt{kn}\log k \min\{k,\log n\})$
and hence the total runtime of the algorithm is 
$O(\sqrt{kn}\cdot n/k \cdot \log k \min\{k,\log n\}\log (1/\epsilon))$ which is
$O(n^{3/2}\log(1/\epsilon))$ for $k$ constant and
$O(n^{3/2}/k^{1/2}\cdot \log^2 n \cdot \log(1/\epsilon))$ for general $k$.
The claim follows using $k\in \Omega(S/\log n)$.\\[-6ex]
\end{sloppypar}
\end{proof}


To complete the proof of Theorem~\ref{EDtradeoff}, we now prove the two
technical lemmas on properties of $f_{x,h}^*(K)$
that the randomization in the hash functions. 
In particular, we first prove Lemma~\ref{closure} showing that for any $x$,
$f_{x,h}^*(K)$ is typically not larger than $2\sqrt{kn}$.

\begin{proof}[Proof of Lemma~\ref{closure}]
We can run an experiment that is equivalent to selecting $f_{x,h}^*(K)$ as
follows:  

\begin{tabbing}
Set $M=I=K=\emptyset$.\\
For \= $count=1$ to $k$ do\\
  \> Choose an element $i\in [n]$ uniformly at random and add it to $K$.\\
  \> While \= $i\notin I$ do\\
\>\>  Add $i$ to $I$\\
\>\>   If \= $x_i\notin M$ then \\
\>\>\> Add $x_i$ to $M$\\
\>\>\> Choose an element $i\in [n]$ uniformly at random.\\
\>\>     else exit while loop and output (``duplicate found'')\\
\>End While\\
End For\\
Output $I$.
\end{tabbing}

Observe that, when each new element $i$ is chosen at random, the probability
that a previous index is found (the while loop exits) is precisely $|I|/n$
and that $|I|$ increases by 1 per step, except for the $k$ steps when $i$ was
already in $I$ or $x_i$ was in $M$.
View each of these random choices of $i$ as a coin-flip with probability
of heads being $|I|/n$. 

%
%

The index set size $|I|$ at the $j$-th random choice is at least $j-k$ as we
will have seen a duplicate in either $I$ or $M$ at most $k$ times.
Hence the probability of exiting the while loop is at least $(j-k)/n$. 
Therefore the expected number of while loop exits that have occurred when the
$t$-th random choice is made is at least $\sum_{j=1}^{t-k} j/n$.  
Consider for a moment our experiment but with the bound on the for loop
removed.
Then, solving for $t$  we get that the minimum number of random choices at
which the expected number of while loop exits is at least $k/(1-\delta)$ is
itself at least
$\lceil (\sqrt{1+8kn/(1-\delta)}-1)/2+k \rceil \le \sqrt{2kn/(1-\delta)} +k+1$.
Let $\delta = 3/4$ and observe that $n > 32k$ implies that
$ \sqrt{8kn/3}  +k +1 \le 2\sqrt{kn}$. 
Our experiment terminates when $\text{count} > k$ and so we can now bound the
number of random choices and hence $|f_{k,n}^*(K)|$.
By treating the event of exiting the while loop as independent coin flips
with the lower bound probability of $(j-k)/n$ at random choice $j$ and
applying the Chernoff inequality, we get that
$\Pr[|f_{k,n}^*(K)|> 2\sqrt{kn}]\le  e^{-9 k/8}$
as required to prove the lemma.
\end{proof}

Finally we prove Lemma~\ref{found-duplicate} showing that any duplicate in $x$
is also somewhat likely to be found in $f_{x,h}^*(K)$.  

\begin{proof}[Proof of Lemma~\ref{found-duplicate}]
We show this first for an $x$ that has a single pair of duplicated values and
then observe how this extends to the general case.
Consider the sequence of indices selected in the experiment in the proof of
Lemma~\ref{closure} on input $x$. 
We define an associated Markov chain with a state $S(d,j,b)$ for
$0\le d< n$, for $0\le j\le k-1$, and for $b\in \{0,1\}$,
where state $S(d,j,b)$ indicates that there have been $d$ distinct indices
selected so far, $j$ pseudo-duplicates, and $b$ is the number of indices from
the duplicated pair selected so far.   In addition, the chain has two
absorbing states, $F$, indicating that a duplicate has been found, and $N$,
indicating that no duplicates have been found at termination.

Observe that the probability in state $S(d,j,b)$ that the next selection is a
pseudo-duplicate is precisely $d/n$; hence it is also the probability of a
transition from $S(d,j,b)$ to $S(d,j+1,b)$ if $j<k-1$, or to $N$ if $j=k-1$.
From state $S(d,j,0)$ there is a $2/n$ chance of selecting one of the duplicated
pair, so this is the probability of the transition from $S(d,j,0)$ to
$S(d+1,j,1)$.  Finally, from state $S(d,j,1)$, the probability that the other
element of the duplicated pair is found is precisely $1/n$ and so this is 
the probability of a transition to $F$.  The remaining transition from state
from $S(d,j,b)$ leads to state $S(d+1,j,b)$ with probability $1-(d+2-b)/n$.
See Figure~\ref{fig:markov}.

\begin{figure}[t]
\begin{minipage}{2.5in}
\begin{center}
\includegraphics[scale = 0.35]{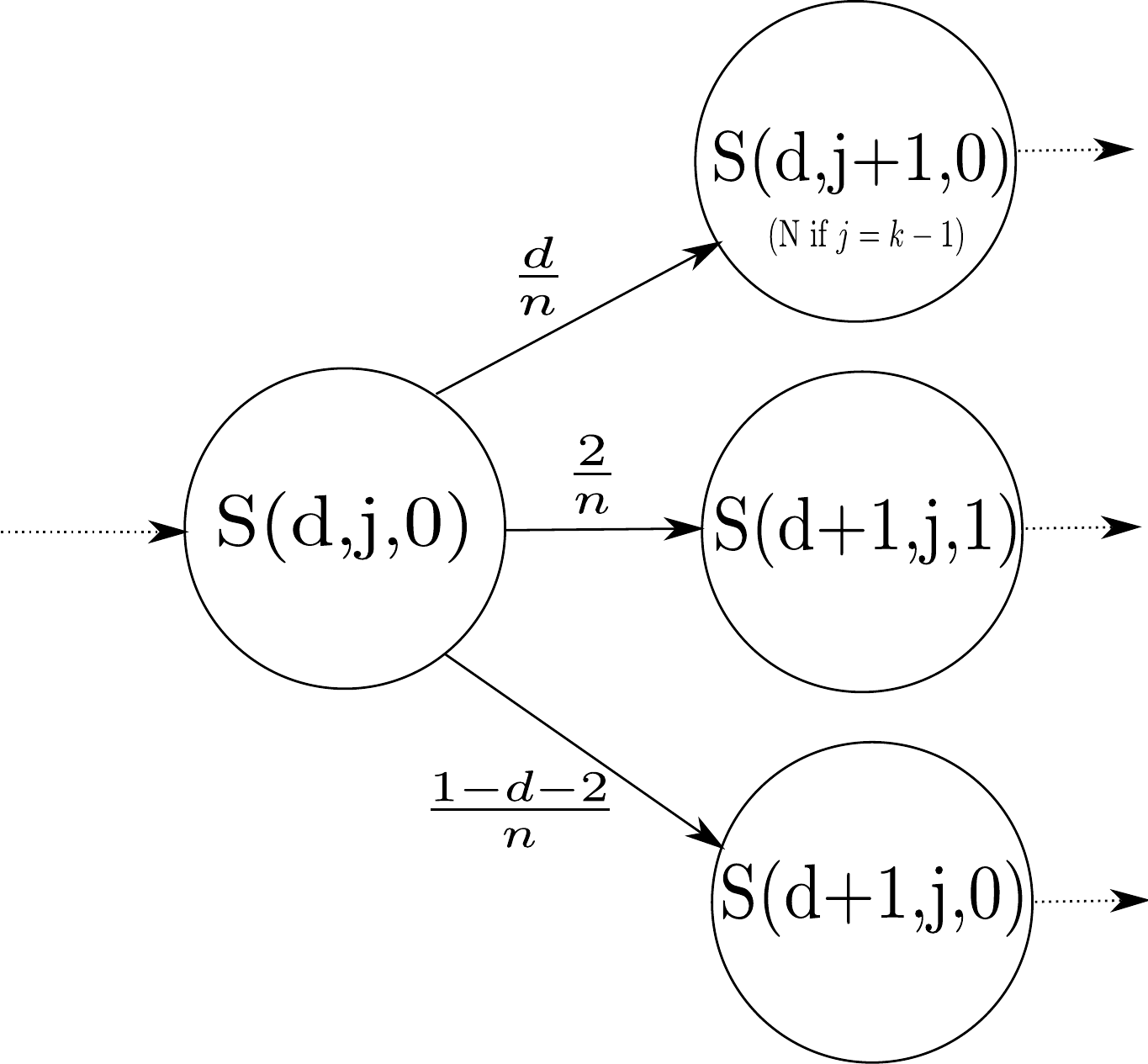}
\end{center}
\end{minipage}
\begin{minipage}{3.25in}
\begin{center}
\includegraphics[scale = 0.35]{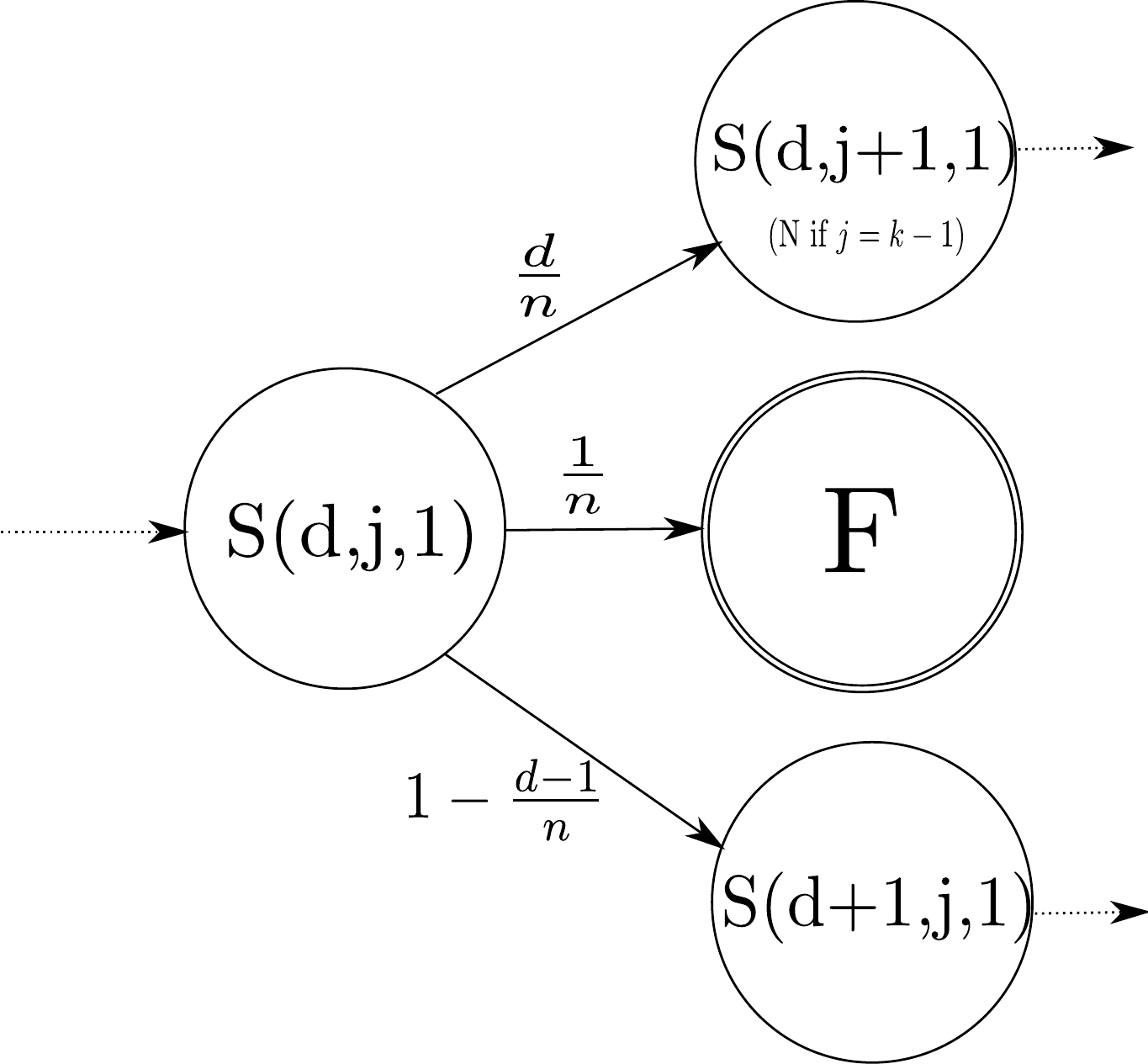}
\end{center}
\end{minipage}
\caption{\small Transitions from the Markov chain in the proof of Lemma~\ref{found-duplicate}}
\label{fig:markov}
\end{figure}

By the result of Lemma~\ref{closure}, with probability at least $8/9$, when
started in state $S(0,0,0)$ this chain
reaches state $F$ or $N$ in at most $t^*=2\sqrt{kn}$ steps.  The quantity
we wish to lower bound is the probability that the chain reaches state $F$ in
its first $t^*$ steps.
We derive our bound by comparing the probability that $F$ is reached in $t^*$
steps to the
probability that either $F$ or $N$ is reached in $t^*$ steps.

Consider all {\em special} transitions, those that increase $j$ or $b$, or
result in states $F$ or $N$.  
In any walk on the Markov chain that reaches $F$ or $N$, at most $k+1$
special transitions can be taken, and to reach $N$, at least
$k$ special transitions must be taken.

Fix any sequence $0\le d_1\le\ldots\le d_{k+1}<t^*$. 
We say that a walk on the Markov chain is
{\em consistent} with this sequence iff the sequence of starting states
of its special transitions is a prefix of some sequence of states
$S(d_1,*,*),\ldots, S(d_{k+1},*,*)$.   We condition on the walk
being one of these consistent walks.

Assume that $k\ge 2$.  
In order to reach state $F$, there must be one special transition that changes
$b$ from 0 to 1 and another that leads to $F$.
Consider the choices $(a,a')\in \binom{[k]}{2}$ of where these can
occur among the first $k$ special transitions.  The conditional probability
that a consistent special transition goes from $S(d_a,j_a,0)$ to
$S(d_a+1,j_a,1)$
is precisely $2/(d_a+2)\ge 1/t^*$ and the conditional probability that
it goes from $S(d_{a'},j_{a'},1)$ to $F$ is precisely $1/(d_{a'}+1)\ge 1/t^*$.
In particular, the conditional probability that two special transitions of these
types occur is lower bounded by the probability that $k$ Bernoulli trials 
with success probability $1/t^*$ yield at least two successes.
This lower bound is at least 
\begin{align*}
\binom{k}{2} 1/(t^*)^2 (1-1/t^*)^{k-2}&\ge 
[1-(k-2)/t^*]\cdot \binom{k}{2}/(t^*)^2\\
& \ge [1-(k-2)/(2\sqrt{kn})] k(k-1)/(8kn) \ge k/(16n).
\end{align*}
Since the different sequences correspond to disjoint sets of inputs,
conditioned on the event that $F$ or $N$ is reached in at most $t^*$ steps,
which occurs with probability at least $8/9$, the conditional
probability that $F$ is reached is $\ge k/(16n)$.  Therefore the
total probability that $F$ is reached in at most $t^*$ steps is
at least $k/(18n)$, as required\footnote{We have not tried to optimize these
constant factors. Also, one can similarly do a separate sharper analysis when
$k=1$.}.   

In the general case,
we observe that at each step prior to termination, the probability of
finding a pseudo-duplicate, given that $j$ of them have previously been found,
does not depend on $x$.  On the other hand, additional duplicated inputs in $x$
only increase the chance of selecting a first index that is duplicated, and
can only increase the chance that one of its matching indices will be selected
in subsequent choices. Therefore, finding a duplicate is at least as likely for
$x$ as it is in the above analysis.
\end{proof}

\section{Sliding Windows}

Let $D$ and $R$ be two finite sets and
$f:D^n\rightarrow R$ be a function over strings of length $n$. We
define the operation $\boxplus$ which
takes $f$ and returns a function $f^{\boxplus t}:D^{n+t-1}\rightarrow R^t$,
defined by
$f^{\boxplus t} (x) = \left(f(x_i\ldots x_{i+n-1})\right)_{i=1}^{t}$.
Because it produces a large number of outputs while less than doubling the
input size, we concentrate on the case that $t=n$ and
apply the $\boxplus$ operator to the statistical 
functions $F_k$, $F_k \bmod 2 $, $ED$, and $O_t$, the $t^{th}$ order statistic.
We will use the notation $F_k^{(j)}$
(resp. $f_i^{(j)}$) to denote the $k^{th}$ frequency moment (resp. the frequency of symbol $i$)
of the string in the window of length $n$ starting at position $j$.

\subsection{Element Distinctness over Sliding Windows}
\label{sec:slidingspacesaving}

The main result of this section shows that our randomized
branching program for $ED_n$ can even be extended
to a $T\in \tilde O(n^{3/2}/S^{1/2})$ randomized branching program for
its sliding windows version $ED_n^{\boxplus n}$.
We do this in two steps.
We first give a deterministic reduction which shows how the answer to
an element distinctness problem allows one to reduce the input size of
sliding-window algorithms for computing $ED_n^{\boxplus m}$.

\begin{lemma}
\label{windowED-size-reduction}
Let $n>m>0$.
\begin{enumerate}[(a)]
\item If $ED_{n-m+1}(x_m,\ldots, x_n)=0$ then
$ED_n^{\boxplus m}(x_1,\ldots, x_{n+m-1})=0^m$.
\item If $ED_{n-m+1}(x_m,\ldots, x_n)=1$ then define
\begin{enumerate}[i.]
\item $i_L=\max\{j\in [m-1]\mid ED_{n-j+1}(x_j,\ldots, x_n)=0\}$ where
$i_L=0$ if the set is empty and
\item $i_R=\min\{j\in [m-1]\mid ED_{n-m+j}(x_m,\ldots, x_{n+j})=0\}$ where $i_R=m$
if the set is empty.
\end{enumerate}
Then
$$ED_n^{\boxplus m}(x_1,\ldots, x_{n+m-1})=0^{i_L} 1^{m-i_L}\ \land\ 1^{i_R} 0^{m-i_R}\ \land\ ED_{m-1}^{\boxplus m}(x_1,\ldots,x_{m-1},x_{n+1},\ldots, x_{n+m-1})$$
where each $\land$ represents bit-wise conjunction.
\end{enumerate}
\end{lemma}

\begin{proof}
The elements $M=(x_m,\ldots, x_n)$ appear in all $m$ of the windows so if this
sequence contains duplicated elements, so do all of the windows and hence the
output for all windows is $0$.  This implies part (a).

If $M$ does not contain any duplicates then any duplicate
in a window must involve at least one element from $L=(x_1,\ldots, x_{m-1})$ or
from $R=(x_{n+1},\ldots, x_{n+m-1})$.
If a window has value 0 because it contains an element of $L$ that also appears in $M$, it must also contain the rightmost such element of $L$ and hence any
window that is distinct must begin to the right of this rightmost such element
of $L$.
Similarly, if a window has value 0 because it contains an element of $R$ that
also appears in $M$, it must also contain the leftmost such element of $R$ and
hence any window that is distinct must end to the left of this leftmost such
element of $R$.
The only remaining duplicates that can occur in a window can only involve
elements of both $L$ and $R$.
In order, the $m$ windows contain the following sequences of elements of $L\cup R$:
$(x_1,\ldots, x_{m-1})$, $(x_2,\ldots, x_{m-1}, x_{n+1})$, $\ldots$, $(x_{m-1},x_{n+1},\ldots, x_{n+m-2})$, $(x_{n+1},\ldots,x_{n+m-1})$.
These are precisely the sequences for which $ED_{m-1}^{\boxplus m}(x_1,\ldots,x_{m-1},x_{n+1},\ldots,x_{n+m-1})$ determines distinctness.   Hence
part (b) follows.
\end{proof}

We use the above reduction in input size to show that any efficient algorithm
for element distinctness can be extended to solve element distinctness over
sliding windows at a small additional cost.

\begin{lemma}
\label{windowED}
If there is an algorithm $A$ that solve element distinctness, $ED$, using
time at most $T(n)$ and
space at most $S(n)$, where $T$ and $S$ are nondecreasing functions of $n$,
then there is an algorithm $A^*$ that solves the sliding-window
version of element distinctness, $ED_n^{\boxplus n}$, in time $T^*(n)$ that is
$O(T(n)\log^2 n)$ and space $S^*(n)$ that is $O(S(n)+\log^2 n)$.  Moreover,
if $T(n)$ is $\Omega(n^\beta)$ for $\beta>1$, then $T^*(n)$ is $O(T(n)\log n)$.

If $A$ is deterministic then so is $A^*$.   If $A$ is randomized with error
at most $\epsilon$ then $A^*$ is randomized with error $o(1/n)$.  Moreover,
if $A$ has the obvious 1-sided error (it only reports that inputs are not
distinct if it is certain of the fact) then the same property holds for $A^*$.
\end{lemma}

\begin{proof}
We first assume that $A$ is deterministic.
Algorithm $A^*$ will compute the $n$ outputs of $ED_n^{\boxplus n}$ in $n/m$
groups of $m$ using
the input size reduction method from Lemma~\ref{windowED-size-reduction}.
In particular, for each group $A^*$ will first call $A$ on the middle section
of input size $n-m+1$ and output $0^m$
if $A$ returns 0.   Otherwise, $A^*$ will do two binary searches involving
at most $2\log m$ calls to $A$ on inputs of size at most $n$ to compute $i_L$
and $i_R$ as defined in part (b) of that lemma.
Finally, in each group, $A^*$ will make
one recursive call to $A^*$ on a problem of size $m$.

It is easy to see that this yields a recurrence of the form
$$T^*(n)=(n/m)[cT(n) \log m + T^*(m)].$$
In particular, if we choose $m=n/2$ then we obtain
$T^*(n)\le 2 T^*(n/2)+ 2c T(n)\log n$.   If $T(n)$ is $\Omega(n^\beta)$ for $\beta>1$
this solves to $T^*(n) \in O(T(n)\log n)$.  Otherwise, it is immediate from
the definition of $T(n)$ that $T(n)$ must be $\Omega(n)$ and hence the recursion
for $A^*$ has $O(\log n)$ levels and the total cost associated with each
of the levels of the recursion is $O(T(n)\log n)$.

Observe that the space for all the calls to $A$ can be re-used in the recursion.
Also note that the algorithm $A^*$ only needs to remember a constant number
of pointers for each level of recursion for a total cost of $O(\log^2 n)$
additional bits.

We now suppose that the algorithm $A$ is randomized with error at most
$\epsilon$.
For the recursion based on Lemma~\ref{windowED-size-reduction}, we use
algorithm $A$ and run it $C=O(\log n)$ times on input $(x_m,\ldots,x_n)$,
taking the majority of the answers to reduce the error to $o(1/n^2)$.
In case that no duplicate is found in these calls, we then apply the noisy
binary search method of Feige, Peleg, Raghavan, and
Upfal~\cite{fpru:noisy-decision} to
determine $i_L$ and $i_R$ with error at most $o(1/n^2)$ by using only
$C=O(\log n)$ calls to $A$.   (If the
original problem size is $n$ we will use the same fixed number $C=O(\log n)$
of calls to $A$ even at deeper levels of the recursion so that each subproblem
has error $o(1/n^2)$.)
There are only $O(n)$ subproblems so the final error is $o(1/n)$.
The rest of the run-time analysis is the same as in the deterministic case.

If $A$ has only has false positives (if it claims that the input is not
distinct then it is certain that there is a duplicate) then observe that
$A^*$ will only have false positives.
\end{proof}

Combining Theorem~\ref{EDtradeoff} with Lemma~\ref{windowED} 
we obtain our algorithm for element distinctness over sliding windows.

\begin{theorem}
For space $S \in [c\log n,n]$, $ED^{\boxplus n}$ can be solved in time $T \in O(n^{3/2} \log^{7/2}{n}/S^{1/2})$ with $1$-sided error probability
$o(1/n)$. 
If the space $S \in O(\log{n})$ then the time is reduced to $T \in O{\left(n^{3/2} \log{n} \right)}$.
\end{theorem}


When the input alphabet is chosen uniformly at random from $[n]$ there exists a much simpler $0$-error
sliding-window algorithm for $ED^{\boxplus n}$ that is efficient on average. The
algorithm runs in $O(n)$ time on average using $O(\log{n})$ bits of space.
By way of contrast, under the same distribution, we prove an average case
time-space lower bound of $\overline T \in \Omega(n^{2}/ \overline S )$ for
$(F_0\bmod 2)^{\boxplus n}$ in Section~\ref{sec:fk}.
The proof is in Appendix~\ref{sec:average}.

\begin{theorem}
\label{average-window-ED}
For input randomly chosen uniformly from $[n]^{2n-1}$,
$ED^{\boxplus n}$ can be solved in average time $\overline{T} \in O(n)$ 
and average space $\overline{S} \in O(\log{n})$.
\end{theorem}

\subsection{Frequency Moments over Sliding Windows}\label{sec:fk}

We now show a $T\in \Omega(n^2/S)$ lower bound for randomized branching
programs computing frequency moments over sliding windows.   
This contrasts with our significantly smaller
$T\in \tilde O(n^{3/2}/S^{1/2})$ upper bound from the previous section for
computing element distinctness over sliding windows in this same model,
hence separating the complexity of $ED$ and $F_k$ for $k\ne 1$ over sliding
windows.  Our lower bound also applies to $F_0\bmod 2$.

\subsubsection{A general sequential lower bound for $F_k^{\boxplus n}$ and $(F_0\bmod 2)^{\boxplus n}$}

We derive a time-space tradeoff lower bound for randomized
branching programs computing $F_k^{\boxplus n}$ for $k=0$ and $k\ge 2$.
Further, we show that the lower bound also holds for computing 
$(F_0\bmod 2)^{\boxplus n}$. (Note that the parity of $F_k$ for $k\ge 1$ is
exactly equal to the parity of $n$; thus the outputs of $(F_k\bmod 2)^{\boxplus n}$ are all equal to $n\bmod 2$.)

\begin{theorem}\label{main_theorem}
Let $k=0$ or $k\ge 2$. There is a constant $\delta>0$ such that any $[n]$-way
branching program of time $T$ and space
$S$ that computes $F_k^{\boxplus n}$ with error at most $\eta$,
$0<\eta<1-2^{-\delta S}$, for input randomly chosen uniformly from $[n]^{2n-1}$ must
have $T\cdot S \in \Omega(n^2)$. The same lower bound holds for $(F_0\bmod 2)^{\boxplus n}$.
\end{theorem}

\begin{corollary}
\label{average-and-random}
Let $k=0$ or $k\ge 2$.
\begin{itemize}
\item The average time $\overline T$ and average space
$\overline S$ needed
to compute $(F_k)^{\boxplus n}(x)$ for $x$ randomly chosen uniformly
from $[n]^{2n-1}$ satisfy $\overline T\cdot \overline S\in \Omega(n^2)$.
\item For $0<\eta<1-2^{-\delta S}$, any $\eta$-error randomized RAM or
word-RAM algorithm
computing $(F_k)^{\boxplus n}$ using time $T$ and space $S$
satisfies $T\cdot S\in\Omega(n^2)$.
\end{itemize}
\end{corollary}

\begin{proof}[Proof of Theorem~\ref{main_theorem}]
We derive the lower bound for $F_0^{\boxplus n}$ first.
Afterwards we show the modifications
needed for $k \geq 2$ and for computing
$(F_0\bmod 2)^{\boxplus n}$.
For convenience, on input $x\in [n]^{2n-1}$, we write $y_i$ for the output
$F_k(x_i,\ldots,x_{i+n-1})$.

We use the general approach
of Borodin and Cook~\cite{bc:sorting} together with the observation
of~\cite{abr:tradeoff} of how it applies to average case complexity and
randomized branching programs.
In particular, we divide the branching program $B$ of length $T$ into layers
of height $q$ each.
Each layer is now a collection of small branching programs $B'$, each of
whose start node is a node at the top level of the layer.
Since the branching program must produce $n$ outputs for each input $x$,
for every input $x$ there exists
a small branching program $B'$ of height $q$ in some layer that produces at
least $nq/T>S$ outputs.
There are at most $2^S$ nodes in $B$ and hence there are at most $2^S$ such
small branching programs among all the layers of $B$.
One would normally prove that the fraction of $x\in [n]^{2n-1}$ for which
any one such small program correctly produces the desired number of outputs
is much smaller than $2^{-S}$ and hence derive the desired lower bound.
Usually this is done by arguing that the fraction of inputs consistent with
any path in such a small branching program for which a fixed set of outputs
is correct is much smaller than $2^{-S}$.

This basic outline is more complicated in our argument.
One issue is that if a path in a small program $B'$ finds that certain
values are equal, then the answers to nearby windows may be strongly correlated
with each other; for example, if $x_i=x_{i+n}$ then $y_i=y_{i+1}$.
Such correlations risk making the likelihood too high that the correct
outputs are produced on a path. Therefore, 
instead of considering the total number of outputs produced, 
we reason about the number of outputs from positions that are not duplicated
in the input and argue that with high probability there will be a linear number
of such positions.

A second issue is that inputs for which
the value of $F_0$ in a window happens to be extreme, say $n$ - all distinct -
or 1 - all identical, allow an almost-certain prediction of the value of
$F_0$ for the next window.
We will use the fact that under the uniform distribution, cases like these
almost surely do not happen; indeed the numbers of
distinct elements in every window almost surely fall in a range close to
their mean and in this case the value in the next window will be
predictable with probability bounded below 1 given the value in the previous
ones.
In this case we use the chain rule to compute the overall probability
of correctness of the outputs.

We start by analyzing the likelihood that an output of $F_0$ is extreme.

\begin{lemma}\label{lemma:f0_range_prob_a}
Let $a$ be chosen uniformly at random from $[n]^{n}$.
Then the probability that $F_0(a)$ is
between $0.5n$ and $0.85n$ is at least $1-2 e^{-n/50}$.
\end{lemma}

\begin{proof}
For $a=a_1\ldots a_n$ uniformly chosen from $[n]^n$,
\[
\mathbb{E}[F_0(a)] = \sum_{\ell\in [n]}
\Pr_a[\exists i \in [n]\mbox{ such that } a_i= \ell] = n [ 1- (1-1/n)^n ].
\]
Hence $0.632n<(1-1/e)n<\mathbb{E}[F_0(a)]\le 0.75 n$.
Define a Doob martingale
$D_t$, $t= 0,1,\ldots, n$ with respect to the sequence $a_1\ldots a_n$
by $D_t = \mathbb{E}[F_0(a)\mid a_1\ldots a_t]$.
Therefore $D_0 = \mathbb{E}[F_0(a)]$ and $D_n = F_0(a)$.
Applying the Azuma-Hoeffding inequality, we have
\[
\Pr_a[F_0(a)\notin [0.5n,0.85n]]\le
\Pr_a[|F_0(a)-\mathbb{E}[F_0(a)| \geq 0.1n]
\le 2 e^{- 2\frac{(0.1n)^2}{n}} = 2 e^{-n/50},\]
which proves the claim.
\end{proof}

We say that $x_j$ is \emph{unique in $x$} iff 
$x_j \notin \{x_1,\ldots,x_{j-1},x_{j+1},\ldots,x_{2n-1}\}$.

\begin{lemma}\label{lemma:f0_range_prob}\label{manyruns}
Let $x$ be chosen uniformly at random from $[n]^{2n-1}$ with $n \geq 2$.
With probability at least $1-4n e^{-n/50}$,
\begin{enumerate}[(a)]
\item all outputs of $F_0^{\boxplus n}(x)$ are between $0.5n$ and $0.85n$,
and
\item the number of positions $j < n$ such that $x_j$ is unique in $x$ is at
least $n/24$.
\end{enumerate}
\end{lemma}

\begin{proof}
We know from Lemma~\ref{lemma:f0_range_prob_a} and the union bound that part (a) is false with probability at most $2 n e^{-n/50}$. 
For any $j<n$, let $U_j$ be the indicator variable of the event that $j$ is
unique in $x$ and $U=\sum_{j<n} U_j$.
Now $\E(U_j)=(1-1/n)^{2n-2}$ so $\E(U)= (n-1) (1-1/n)^{2n-2} \geq n/8$ for
$n\ge 2$.
Observe also that this is a kind of typical ``balls in bins" problem 
and so, as discussed in~\cite{dp:concentration-book}, it has the
property that the random variables $U_j$ are \emph{negatively
associated}; for example, for disjoint $A,A'\subset[n-1]$,  the larger 
$\sum_{j\in A} U_j$ is, the smaller $\sum_{j\in A'} U_j$ is likely to be.
Hence, it follows~\cite{dp:concentration-book}
that $U$ is more closely concentrated around its mean than if
the $U_j$ were fully independent.
It also therefore follows that we can apply a Chernoff bound directly to our
problem, giving
$\Pr[U \leq n/24] \leq \Pr[U \leq \E(U)/3] \leq e^{-2\E(U)/9} \leq e^{-n/36}$.
We obtain the desired bound for parts (a) and (b) together by another
application of the union bound.
\end{proof}

\paragraph{Correctness of a small branching program for computing outputs
in $\pi$-unique positions}

\begin{definition}
Let $B'$ be an $[n]$-way branching program and let $\pi$ be a source-sink
path in $B'$ with queries $Q_\pi$ and answers $A_\pi:Q_\pi\rightarrow [n]$.  
An index $\ell<n$ is said to be \emph{$\pi$-unique} iff 
either (a) $\ell \notin Q_\pi$, or 
(b) $A_\pi(\ell) \notin A_\pi(Q_\pi-\{\ell\})$.
\end{definition}

In order to measure the correctness of a small branching program, we 
restrict our attention to outputs that are produced at positions that are
$\pi$-unique and upper-bound the probability that a small branching program
correctly computes outputs of $F_0^{\boxplus n}$ at many $\pi$-unique positions
in the input.

Let $\mathcal{E}$ be the event that all outputs of $F_0^{\boxplus n}(x)$
are between $0.5n$ and $0.85n$.

\begin{lemma}
\label{runs-lemma}
Let $r> 0$ be a positive integer, let $\epsilon \le 1/10$, and let
$B'$ be an $[n]$-way branching program of height $q=\epsilon n$.
Let $\pi$ be a path in $B'$ on which
outputs from at least $r$ $\pi$-unique positions are produced.
For random $x$ uniformly chosen from $[n]^{2n-1}$,
\[
\Pr[\mbox{these $r$ outputs are correct for $F_0^{\boxplus n}(x)$},\mathcal{E}
\mid \pi_{B'}(x)=\pi]\le (17/18)^r.
\]
\end{lemma}

\begin{proof}
Roughly, we will show that when $\mathcal{E}$ holds (outputs for all
windows are not extreme) then, conditioned
on following any path $\pi$ in $B'$, each output
produced for a $\pi$-unique position will have only a constant probability
of success conditioned on any outcome for the previous outputs.  Because of the
way outputs are indexed, it will be convenient to consider these outputs
in right-to-left order.

Let $\pi$ be a path in $B'$, $Q_\pi$ be the set of queries along $\pi$,
$A_\pi:Q_\pi\rightarrow [n]$ be the answers along $\pi$, and
$Z_\pi:[n]\rightarrow [n]$ be the partial function denoting the outputs 
produced along $\pi$.
Note that $\pi_{B'}(x)=\pi$ if and only if $x_i=A_\pi(i)$ for all $i\in Q_\pi$.

Let $1\le i_1<i_2<\ldots < i_r<n$ be the first $r$ of the $\pi$-unique
positions on which $\pi$ produces output values; i.e.,
$\{i_1,\ldots, i_r\}\subseteq \mathrm{dom}(Z_\pi)$.
Define $z_{i_1}=Z_\pi(i_1),\ldots, z_{i_r}=Z_\pi(i_r)$.

We will decompose the probability over the input $x$ that $\mathcal{E}$ and all of $y_{i_1}=z_{i_1},\ldots, y_{i_r}=z_{i_r}$ hold via the chain rule.   
In order to do so,
for $\ell \in [r]$, we define event 
$\mathcal{E}_\ell$ to be $0.5n \le F_0^{(i)}(x)\le 0.85n$ for all $i>i_{\ell}$.
We also write $\mathcal{E}_{0} \myeq \mathcal{E}$.  Then
\begin{align}
&\Pr[y_{i_1}=z_{i_1},\ldots,y_{i_r}=z_{i_r},\ \mathcal{E}
\mid \pi_{B'}(x)=\pi]\nonumber\\
& = \Pr[\mathcal{E}_r\mid \pi_{B'}(x)=\pi]\cdot \prod_{\ell = 1}^{r} \Pr [y_{i_\ell}=z_{i_\ell},\ \mathcal{E}_{\ell-1}\mid
y_{i_{\ell+1}}= z_{i_{\ell+1}},\ \ldots,\ y_{i_{r}}= z_{i_{r}},\ \mathcal{E}_\ell,
\ \pi_{B'}(x)=\pi]\nonumber\\
& \le \prod_{\ell = 1}^{r} \Pr [y_{i_\ell}=z_{i_\ell}\mid
y_{i_{\ell+1}}= z_{i_{\ell+1}},\ \ldots,\ y_{i_{r}}= z_{i_{r}},\ \mathcal{E}_\ell,
\ \pi_{B'}(x)=\pi].
\label{chain rule}
\end{align}

We now upper bound each term in the product in~\eqref{chain rule}.
Depending on how much larger $i_{\ell+1}$ is than $i_\ell$, the conditioning
on the value of $y_{i_{\ell+1}}$ may imply a lot of information about
the value of $y_{i_\ell}$, but we will show that even if we reveal more
about the input, the value of $y_{i_\ell}$ will still have a constant amount
of uncertainty.

For $i\in [n]$, let $W_i$ denote the vector of input elements
$(x_i,\ldots,x_{i+n-1})$, and note that $y_i=F_0(W_i)$; we call $W_i$ the
$i^\mathrm{th}$ window of $x$.
The values $y_i$ for different windows may be closely related.
In particular, adjacent windows $W_i$ and $W_{i+1}$ have numbers of distinct
elements that can differ by at most 1 and this depends on whether the
extreme end-points of the two windows, $x_i$ and $x_{i+n}$, appear among
their common elements $C_i=\{x_{i+1},\ldots,x_{i+n-1}\}$.
More precisely,
\begin{equation} 
y_{i} - y_{i+1}=
\textbf{1}_{\{x_{i} \not \in C_i\}}
-
\textbf{1}_{\{x_{i+n} \not \in C_i\}}.
\label{indicators}
\end{equation}
In light of \eqref{indicators}, the basic idea of our argument is that,
because $i_\ell$ is $\pi$-unique and because of the conditioning on
$\mathcal{E}_\ell$, there will be enough uncertainty about
whether or not
$x_{i_\ell} \in C_{i_\ell}$  to show that the value of
$y_{i_\ell}$ is uncertain even if we reveal 
\begin{enumerate}
\item the value of the indicator
$\textbf{1}_{\{x_{i_{\ell}+n} \not \in C_{i_{\ell}}\}}$, and
\item the value of the output $y_{i_{\ell}+1}$.
\end{enumerate}

We now make this idea precise in bounding each term in the product
in~\eqref{chain rule}, using $\mathcal{G}_{\ell+1}$ to denote the event 
$\{y_{i_{\ell+1}}= z_{i_{\ell+1}},\ \ldots,\ y_{i_{r}}= z_{i_{r}}\}$.
\begin{align}
\Pr &[y_{i_\ell}=z_{i_\ell}\mid \mathcal{G}_{\ell+1},\ \mathcal{E}_\ell,
\ \pi_{B'}(x)=\pi]\notag\\
=&\sum_{m=1}^n \sum_{b \in \{0,1\}}
\Pr [y_{i_\ell}= z_{i_\ell}\mid
y_{i_{\ell} +1} = m, \textbf{1}_{\{x_{i_{\ell}+n} \not \in C_{i_{\ell}}\}} =b,\mathcal{G}_{\ell+1}, \mathcal{E}_\ell,\pi_{B'}(x)=\pi]\notag\\
& \qquad\times \Pr[y_{i_{\ell} +1} = m , \textbf{1}_{\{x_{i_{\ell}+n} \not \in C_{i_{\ell}}\}} =b\mid \mathcal{G}_{\ell+1}, \mathcal{E}_\ell,\pi_{B'}(x)=\pi]\notag\\
\leq & \max_{\substack{m \in [0.5n,0.85n]\\ b\in \{0,1\}}}
\Pr [y_{i_\ell}= z_{i_\ell}\mid
y_{i_{\ell} +1} = m, \textbf{1}_{\{x_{i_{\ell}+n} \not \in C_{i_{\ell}}\}} =b,
\mathcal{G}_{\ell+1}, \mathcal{E}_\ell,\pi_{B'}(x)=\pi]\notag\\
= & \max_{\substack{m \in [0.5n,0.85n]\\ b\in \{0,1\}}}
\Pr [ \textbf{1}_{\{x_{i_{\ell}} \not \in C_{i_{\ell}}\}}
= z_{i_\ell} -m+b\mid \notag\\[-3ex]
&\qquad\qquad\qquad\qquad y_{i_{\ell} +1} = m,\textbf{1}_{\{x_{i_{\ell}+n} \not \in C_{i_{\ell}}\}} =b,
\mathcal{G}_{\ell+1}, \mathcal{E}_\ell,\pi_{B'}(x)=\pi]
\label{correct-diff}
\end{align}
where the inequality follows because the conditioning
on $\mathcal{E}_\ell$ implies that
$y_{i_\ell +1}$ is between $0.5n$ and $0.85n$ and
the last equality follows because of the conditioning together with
\eqref{indicators} applied with $i=i_\ell$.  
Obviously, unless $z_{i_\ell}-m+b\in \{0,1\}$ the probability of the
corresponding in the maximum in \eqref{correct-diff} will be 0.
We will derive our bound by showing that given all the conditioning in
\eqref{correct-diff}, the probability of the event
$\{x_{i_{\ell}} \not \in C_{i_{\ell}}\}$ is between $2/5$ and $17/18$
and hence each term in the product in \eqref{chain rule} is at most $17/18$.

\paragraph{Membership of $x_{i_\ell}$ in $C_{i_\ell}$:} 
First note that the conditions $y_{i_\ell+1}=m$ and 
$\textbf{1}_{\{x_{i_{\ell}+n} \not \in C_{i_{\ell}}\}} =b$ together
imply that $C_{i_{\ell}}$ contains precisely $m-b$ distinct values.
We now use the fact that $i_\ell$ is $\pi$-unique and, hence,
either $i_\ell\notin Q_\pi$ or $A_\pi(i_\ell)\notin A_\pi(Q_\pi-\{i_\ell\})$.

First consider the case that $i_\ell\notin Q_\pi$.  
By definition, the events $y_{i_\ell+1}=m$, 
$\textbf{1}_{\{x_{i_{\ell}+n} \not \in C_{i_{\ell}}\}} =b$,
$\mathcal{E}_\ell$, and
$\mathcal{G}_{\ell+1}$ only depend on $x_i$ for $i>i_\ell$ and 
the conditioning on $\pi_{B'}(x)=\pi$ is only a property of
$x_i$ for $i\in Q_\pi$.   Therefore, under all the conditioning in
\eqref{correct-diff},
$x_{i_\ell}$ is still a uniformly random value in $[n]$.
Therefore the probability that 
$x_{i_\ell}\in C_{i_\ell}$ is precisely $(m-b)/n$ in this case.

Now assume that $i_\ell\in Q_\pi$.  In this case, the conditioning
on $\pi_{B'}(x)=\pi$ implies that $x_{i_\ell}=A_\pi(i_\ell)$ is fixed
and not in $A_\pi(Q_\pi -\{i_\ell\})$.
Again, from the conditioning we know that $C_{i_\ell}$ contains precisely
$m-b$ distinct values.  
Some of the elements that occur in $C_{i_\ell}$ may be inferred from the
conditioning -- for example, their values may have been queried along $\pi$ --
but we will show that there is significant uncertainty about whether any
of them equals $A_\pi(i_\ell)$.
In this case we will show that the uncertainty persists even if we reveal
(condition on) the locations of
all occurrences of the elements $A_\pi(Q_\pi -\{i_\ell\})$ among the $x_i$
for $i>i_\ell$.

Other than the information revealed about the occurrences of the elements
$A_\pi(Q_\pi -\{i_\ell\})$ among the $x_i$ for $i>i_\ell$,
the conditioning on the events $y_{i_\ell+1}=m$, 
$\textbf{1}_{\{x_{i_{\ell}+n} \not \in C_{i_{\ell}}\}} =b$,
$\mathcal{E}_\ell$, and
$\mathcal{G}_{\ell+1}$,
only biases the numbers of distinct elements and patterns of equality among
inputs $x_i$ for $i>i_\ell$.
Further the conditioning on $\pi_{B'(x)}=\pi$ does not reveal anything
more about the inputs in $C_{i_\ell}$ than is given by the occurrences of
$A_\pi(Q_\pi -\{i_\ell\})$.
Let $\mathcal{A}$ be the event that all the conditioning is true.

Let $q'=|A_\pi(Q_\pi -\{i_\ell\})|\le q-1$
and let $q''\le q'$ be the number of distinct elements of
$A_\pi(Q_\pi -\{i_\ell\})$ that appear in $C_{i_\ell}$.
Therefore, since the input is uniformly chosen,  subject to the conditioning,
there are $m-b-q''$ distinct elements
of $C_{i_\ell}$ not among $A_\pi(Q_\pi -\{i_\ell\})$, and these distinct
elements are
uniformly chosen from among the elements $[n]-A_\pi(Q_\pi-\{i_\ell\})$.
Therefore, the probability that any of these $m-b-q''$ elements is
equal to $x_{i_\ell}=A_\pi(i_\ell)$ is precisely $(m-b-q'')/(n-q')$ in this
case.

It remains to analyze the extreme cases of 
the probabilities $(m-b)/n$ and $(m-b-q'')/(n-q')$ from the discussion above.
Since $q=\epsilon n$, $q''\le q'\le q-1$, and $b\in \{0,1\}$, we have the
probability
$\Pr[x_{i_\ell} \in C_{i_\ell}\mid \mathcal{A}] \leq \frac{m}{n - q+1} \leq
\frac{0.85n}{n- \epsilon n} \leq
\frac{0.85n}{n(1-\epsilon)} \leq  0.85/(1-\epsilon)\le 17/18$ since
$\epsilon\le 1/10$.
Similarly,
$\Pr[x_{i_\ell} \notin C_{i_\ell}\mid \mathcal{A}] 
< 1- \frac{m-q}{n}\leq 1- \frac{ 0.5 n- \epsilon n}{n} \leq 0.5+\epsilon\le 3/5$ since $\epsilon \le 1/10$.
Plugging in the larger of these upper bounds in~\eqref{chain rule},
we get:
\[
\Pr [z_{i_1}, \ldots, z_{i_r} \mbox{ are correct for }
F_0^{\boxplus n}(x) ,\ \mathcal{E}\mid \pi_{B'}(x)=\pi]
\leq (17/18)^{r},
\]
which proves the lemma.
\end{proof}

\paragraph{Putting the Pieces Together}
We now combine the above lemmas.
Suppose that $T S \le n^2/4800$ and let $q=n/10$.
We can assume without loss of generality that $S\ge \log_2 n$ since
we need $T\ge n$ to determine even a single answer.

Consider the fraction of inputs in $[n]^{2n-1}$ on which $B$ correctly computes
$F_0^{\boxplus n}$.
By Lemma~\ref{manyruns},
for input $x$ chosen uniformly from $[n]^{2n-1}$, the probability that
$\mathcal{E}$ holds and there are at least
$n/24$ positions $j<n$ such that $x_j$ is unique in $x$
is at least $1-4ne^{-n/50}$.
Therefore, in order to be correct on any such $x$, $B$ must correctly
produce outputs from at least $n/24$ outputs at positions $j<n$ such that 
$x_j$ is unique in $x$.

For every such input $x$, by our earlier outline, one of the $2^S$
$[n]$-way branching programs $B'$ of height $q$ contained in $B$ produces
correct output values for $F_0^{\boxplus n}(x)$  in at least
$r=(n/24)q/T\ge 20S$ positions $j<n$ such that $x_j$ is unique in $x$.

We now note that for any $B'$, if $\pi=\pi_{B'}(x)$ then the fact
that $x_j$ for $j<n$ is unique in $x$ implies that $j$ must be $\pi$-unique.
Therefore, for all but a $4ne^{-n/50}$ fraction of inputs $x$
on which $B$ is correct, $\mathcal{E}$ holds for $x$ and there is one of the
$\le 2^S$ branching programs
$B'$ in $B$ of height $q$ such that the path $\pi=\pi_{B'}(x)$ produces
at least $20S$ outputs at $\pi$-unique positions 
that are correct for $x$.

Consider a single such program $B'$. By Lemma~\ref{runs-lemma}
for any path $\pi$ in $B'$, the fraction of inputs
$x$ such that $\pi_{B'}(x)=\pi$ for which $20S$ of these outputs are
correct for $x$ and produced at $\pi$-unique positions, and
$\mathcal{E}$ holds for $x$ is at most $(17/18)^{20S}< 3^{-S}$.
By Proposition~\ref{manyruns}, this same bound applies to the fraction
of all inputs $x$ with $\pi_{B'}(x)=\pi$ for which $20S$ of these outputs
are correct from $x$ and produced at $\pi$-unique positions, and
$\mathcal{E}$ holds for $x$ is at most $(17/18)^{20S}< 3^{-S}$.

Since the inputs following different paths in $B'$ are disjoint,
the fraction of all inputs $x$ for which $\mathcal{E}$ holds
and which follow some path in $B'$ that yields at least
$20S$ correct answers from distinct runs of $x$ is less than $3^{-S}$.
Since there are at most $2^S$ such height $q$ branching programs,
one of which must produce $20S$ correct outputs from distinct runs of $x$
for every remaining input,
in total only a $2^S 3^{-S}=(2/3)^S$ fraction of all inputs have these
outputs correctly produced.

In particular this implies that $B$ is correct on at most a
$4ne^{-n/50}+(2/3)^S$
fraction of inputs.
For $n$ sufficiently large
this is smaller than $1-\eta$ for any $\eta<1-2^{-\delta S}$ for some
$\delta>0$, which contradicts our original assumption.  This completes the proof of Theorem~\ref{main_theorem}.
\end{proof}

\paragraph{Lower bound for $ (F_0\bmod 2)^{\boxplus n}$}
We describe how to modify the proof of Theorem~\ref{main_theorem} for
computing $F_0^{\boxplus n}$
to derive the same lower bound for computing $(F_0\bmod 2)^{\boxplus n}$.
The only difference is in the proof of Lemma~\ref{runs-lemma}.
In this case, each output $y_i$ is $F_0(W_i)\bmod 2$ rather than $F_0(W_i)$
and \eqref{indicators} is replaced by
\begin{equation} 
y_{i} = 
(y_{i+1}+ \textbf{1}_{\{x_{i} \not \in C_i\}}
-
\textbf{1}_{\{x_{i+n} \not \in C_i\}})\bmod 2.
\end{equation}
The extra information revealed (conditioned on) will be the same as in the
case for $F_0^{\boxplus n}$ but, because the meaning of $y_i$ has changed,
the notation $y_{i_\ell+1}=m$ is replaced by $F_0(W_{i_\ell+1})=m$, $y_{i_\ell+1}$ is then $m\bmod 2$, and the
upper bound in \eqref{correct-diff} is replaced by
\begin{align*}
\max_{\substack{m \in [0.5n,0.85n]\\ b\in \{0,1\}}}
\Pr [ \textbf{1}_{\{x_{i_{\ell}} \not \in C_{i_{\ell}}\}}
=&(z_{i_\ell} -m+b)\bmod 2\mid \notag\\[-3ex]
&F_0(W_{i_{\ell} +1}) = m,\textbf{1}_{\{x_{i_{\ell}+n} \not \in C_{i_{\ell}}\}} =b,
\mathcal{G}_{\ell+1}, \mathcal{E}_\ell,\pi_{B'}(x)=\pi]
\end{align*}
The uncertain event is exactly the same as before, namely whether or not
$x_{i_{\ell}} \in C_{i_{\ell}}$ and the conditioning is essentially exactly
the same, yielding an upper bound of $17/18$.
Therefore the analogue of Lemma~\ref{runs-lemma} also holds for
$(F_0\bmod 2)^{\boxplus n}$ and hence the time-space tradeoff of
$T\cdot S\in \Omega(n^2)$  follows as before.

\paragraph{Lower Bound for  $F_k^{\boxplus n}$, $k\geq 2$}
We describe how to modify the proof of Theorem~\ref{main_theorem} for
computing $F_0^{\boxplus n}$
to derive the same lower bound for computing $F_k^{\boxplus n}$ for $k\ge 2$.
Again, the only difference is in the proof of Lemma~\ref{runs-lemma}.
The main change from the case of $F_0^{\boxplus n}$
is that we need to replace \eqref{indicators} relating the values of consecutive
outputs.
For $k\ge 2$,
recalling that $f^{(i)}_j$ is the frequency of symbol $j$ in window $W_i$, we
now have
\begin{equation}
\label{indicators-k}
 y_{i} - y_{i+1} =\left[\left(f^{(i)}_{x_i}\right)^k - \left(f^{(i)}_{x_i} - 1 \right)^k \right] - \left[\left(f^{(i+1)}_{x_{i+n}} \right)^k - \left(f^{(i+1)}_{x_{i+n}} -1 \right)^k\right].
\end{equation}
We follow the same outline as in the case $k=0$ in order to bound the
probability that $y_{i_\ell}=z_{i_\ell}$
but we reveal the following information, which is somewhat more
than in the $k = 0$ case: 
\begin{enumerate}
\item $y_{i_\ell+1}$, the value of the output immediately after $y_{i_\ell}$,
\item $F_{0}(W_{i_\ell+1})$,
the number of distinct elements in $W_{i_\ell+1}$, and
\item $f^{(i_\ell+1)}_{x_{i_\ell+n}}$,
the frequency of $x_{i_\ell+n}$ in  $W_{i_\ell+1}$.
\end{enumerate}
For $M\in \mathbb{N}$, $m\in [n]$ and $1\le f\le m$,
define $\mathcal{C}_{M,m,f}$ be the event that
$y_{i_\ell+1}=M$, 
$F_{0}(W_{i_\ell+1})=m$, and
$f^{(i_\ell+1)}_{x_{i_\ell+n}}=f$.  Note that $\mathcal{C}_{M,m,f}$ only
depends on the values in $W_{i_\ell+1}$, as was the case for the information
revealed in the case $k=0$.
As before we can then upper bound the $\ell^\mathrm{th}$ term in the
product given in \eqref{chain rule} by
\begin{equation}
\label{upper-k}
\max_{\substack{m \in [0.5n,0.85n]\\ M\in \mathbb{N},f\in [m]}}
\Pr [y_{i_\ell}= z_{i_\ell}\mid
\mathcal{C}_{M,m,f},
\mathcal{G}_{\ell+1}, \mathcal{E}_\ell,\pi_{B'}(x)=\pi]
\end{equation}
Now, by \eqref{indicators-k}, given event $\mathcal{C}_{M,m,f}$, we have
$y_{i_\ell}=z_{i_\ell}$ if and only if
$z_{i_\ell} - M =\left[\left(f^{(i_\ell)}_{x_{i_\ell}}\right)^k - \left(f^{(i_\ell)}_{x_{i_\ell}} - 1 \right)^k \right] - \left[f^k - (f-1)^k\right]$, which we can express as
as a constraint on its only free parameter $f^{(i)}_{x_i}$,
$$\left(f^{(i_\ell)}_{x_{i_\ell}}\right)^k - \left(f^{(i_\ell)}_{x_{i_\ell}} - 1 \right)^k 
=z_{i_\ell} - M - f^k + (f-1)^k.$$  
Observe that this constraint can be satisfied for at most one positive
integer value of $f^{(i_\ell)}_{x_{i_\ell}}$ and that, by definition,
$f^{(i_\ell)}_{x_{i_\ell}}\ge 1$.
Note that $f^{(i_\ell)}_{x_{i_\ell}}=1$ if and only if
$x_{i_\ell}\notin C_{i_\ell}$, where $C_{i_\ell}$ is defined as in the case
$k=0$.
The probability that 
$f^{(i_\ell)}_{x_{i_\ell}}$ takes on a particular value is at most the larger
of the probability that $f^{(i_\ell)}_{x_{i_\ell}}=1$ or that
$f^{(i_\ell)}_{x_{i_\ell}}> 1$
and hence \eqref{upper-k} is at most
\begin{equation*}
\max_{\substack{m \in [0.5n,0.85n]\\ M\in \mathbb{N},f\in [m],c\in \{0,1\}}}
\Pr [ \textbf{1}_{\{x_{i_{\ell}} \not \in C_{i_{\ell}}\}}=c
\mid \mathcal{C}_{M,m,f},
\mathcal{G}_{\ell+1}, \mathcal{E}_\ell,\pi_{B'}(x)=\pi]
\end{equation*}
We now can apply similar reasoning to the $k=0$ case to argue that this
is at most $17/18$:  The only difference is that 
$\mathcal{C}_{M,m,f}$ replaces the conditions 
$y_{i_\ell+1}=F_0(W_{i_\ell+1})=m$ and
$\textbf{1}_{\{x_{i_{\ell}+n} \not \in C_{i_{\ell}}\}} =b$.
It is not hard to see that the same reasoning still applies with the new
condition.  The rest of the proof follows as before.

\subsubsection{A time-space efficient algorithm for $F_k^{\boxplus n}$}

We now show that our time-space tradeoff lower bound for
$F_k^{\boxplus n}$ is nearly
optimal even for restricted RAM models.

\begin{theorem}
\label{classical-Fk}
There is a comparison-based deterministic RAM algorithm for computing
$F^{\boxplus n}_k$ for any fixed integer $k \geq 0$ with time-space tradeoff
$T\cdot S\in O(n^2\log^2{n})$ for all space bounds $S$ with $\log n\le S\le n$.
\end{theorem}

\begin{proof}
We denote the $i$-th output by $y_i=F_k(x_i,\ldots, x_{i+n-1})$.
We first compute $y_1$ using the
comparison-based $O(n^2/S)$ time sorting
algorithm of Pagter and Rauhe~\cite{pr:comparison-sorting}.
This algorithm produces the list of outputs in order by building a space $S$
data structure $D$ over the $n$ inputs and then repeatedly removing and
returning the index of the smallest element from that structure using a {\sc POP} operation.
We perform {\sc POP} operations on $D$ and keep track of the last index popped.
We also will maintain the index $i$ of the previous symbol seen as well as a
counter that tells us the number of times the symbol has been seen so far.
When a new index $j$ is popped, we compare the symbol at that index with the
symbol at the saved index.
If they are equal, the counter is incremented.
Otherwise, we save the new index $j$, update the running total for $F_k$
using the $k$-th power of the counter just computed, and then reset that
counter to 1.

Let $S'=S/\log_2{n}$.
We compute the remaining outputs in $n/S'$ groups of $S'$ outputs at a time.
In particular, suppose that we have already computed $y_i$.
We compute $y_{i+1},\ldots,y_{i+S'}$ as follows:

We first build a single binary search tree for both $x_i,\ldots, x_{i+S'-1}$
and for $x_{i+n},\ldots, x_{i+n+S'-1}$ and include a pointer $p(j)$ from each
index $j$ to the leaf node it is associated with.
We call the elements $x_i,\ldots, x_{i+S'-1}$ the old elements and add them
starting from $x_{i+S'-1}$.
While doing so we maintain a counter $c_j$ for each index $j\in [i,i+S'-1]$ of
the number of times that $x_j$ appears to its right in $x_i,\ldots,x_{i+S'-1}$.
We do the same for $x_{i+n},\ldots, x_{i+n+S'-1}$, which we call the
new elements, but starting from the left.
For both sets of symbols, we also add the list of indices where each element
occurs to the relevant leaf in the binary search tree.

We then scan the $n-S'$ elements $x_{i+S'},\ldots, x_{i+n-1}$ and maintain
a counter $C(\ell)$ at each leaf $\ell$ of each tree to record the number of
times that the element has appeared.

For $j\in [i,i+S'-1]$ we produce $y_{j+1}$ from $y_j$.
If $x_j = x_{j+n}$ then $y_{j+1} = y_j$.
Otherwise, we can use the number of times the old symbol $x_j$ and the
new symbol $x_{j+n}$ occur in the window $x_{j+1},\dots,x_{j+n-1}$ to
give us $y_{j+1}$.
To compute the number of times $x_j$ occurs in the window, we look at
the current head pointer in the new element list associated with leaf $p(j)$
of the binary search tree.
Repeatedly move that pointer to the right if the next position in the list
of that position is at most $n+j-1$.
Call the new head position index $\ell$.
The number of occurrences of $x_j$  in $x_{j+1}, \ldots, x_{S'}$ and
$x_{n+1}, \ldots, x_{n+j}$ is now $c_j+c_\ell$.
The head pointer never moves backwards and so the total number of pointer
moves will be bounded by the number of new elements.
We can similarly compute the number of times  $x_{j+n}$ occurs in the
window by looking at the current head pointer in the old element list
associated with $p(j+n)$ and moving the pointer to the left until it is
at position no less than $j+1$.
Call the new head position in the old element list $\ell'$.

Finally, for $k>0$ we can output $y_{j+1}$ by subtracting
$(1+c_j+c_\ell +C(p(j)))^k - (c_j+c_\ell +C(p(j))^k$ from $y_j$ and adding
$(1+c_{j+n}+c_{\ell'} +C(p(j+n)))^k - (c_{j+n}+c_{\ell'} +C(p(j+n))^k$.
When $k=0$ we compute $y_{j+1}$ by subtracting the value of the indicator
$\textbf{1}_{c_j+c_\ell +C(p(j)) =0}$ from $y_j$ and adding $\textbf{1}_{c_{j+n}+c_{\ell'} +C(p(j+n)) =0 }$.

The total storage required for the search trees and pointers is $O(S'\log n)$
which is $O(S)$.
The total time to compute $y_{i+1},\ldots, y_{i+S'}$ is dominated by the
$n-S'$ increments of counters using the binary search tree, which is
$O(n\log S')$ and hence $O(n\log S)$ time. This computation must be done $(n-1)/S'$
times for a total of $O(\frac{n^2\log S}{S'})$ time.
Since $S'=S/\log n$, the total time including that to compute $y_1$ is
$O(\frac{n^2 \log n\log S}{S})$ and hence $T\cdot S\in O(n^2\log^2 n)$.
\end{proof}

\subsubsection{A time-space efficient quantum algorithm for
$F_0^{\boxplus n}$ }

We now consider the efficiency of quantum computers computing
frequency moments over sliding windows, in particular whether   
our $T\in \tilde\Theta(n^2/S)$ bound for computing $F_0^{\boxplus n}$ extends
to quantum computation.
We show that this is not the case: quantum algorithms are
significantly more efficient for computing $F_0^{\boxplus n}$ than classical
algorithms are.
The efficiency improvement is due to ideas embodied in Grover's quantum search
algorithm~\cite{grover:quant-search}, both directly and in its extension
to a time-space efficient quantum sorting algorithm due to
Klauck~\cite{klauck:quant-sorting}.
We use these algorithms as black boxes without regard to the details
of quantum computation, so we omit more formal discussions of both quantum
computing and how they work.

\begin{theorem}
\label{quantum-f0}
There is a quantum algorithm for computing
$F^{\boxplus n}_0$ with time-space tradeoff
$T\in O(n^{3/2}\log^{3/2}{n}/S^{1/2})$ for all space bounds $S$ with $\log^3 n\le S\le n/\log n$.
\end{theorem}

\begin{proof}
To compute the first output $y_1$ of $F_0^{\boxplus n}$, we run Klauck's
quantum sorting algorithm~\cite{klauck:quant-sorting} in
time $O(n^{3/2}\log^{3/2}{n}/S^{1/2})$.
We let $S'=S/\log n$ and
follow the same idea as in the proof of Theorem~\ref{classical-Fk} and
determine the remaining outputs in $n/S'$ blocks of $S'$ consecutive values
at a time. 

We can divide up the input positions on which the output values depend into
two groups, $O(S')$ {\em boundary} elements that are associated with some but
not all of the $S'$ output values and $n-S'+1$ {\em common} elements that are
associated with all of the $S'$ output values in the block.
As before, we can determine the values of the outputs in a block given (1) the
output value immediately prior to the block, (2) the pattern of appearances
within the boundary elements, and (3) which of the $O(S')$ boundary elements
appear in the common vector of $n-S'+1$ elements.
As before we build a binary search tree for the $O(S')$ boundary elements
as in the classical algorithm which allows us to handle the pattern of
appearances within the boundary elements and has very low cost.
The cost is dominated by that of determining which boundary elements appear in
the common vector of elements. 

We use a standard multi-target variant of Grover's algorithm~(e.g.\cite{bbht:tight-qsearch,gr:quantumsearch}).
In particular, we begin by searching for some member of the common vector
of elements.   Let $b\le S'$ be the number of different boundary elements 
that actually appear in the common vector.  With the binary search tree it
costs only $O(\log S')$ time to test if a value equals one of the binary
elements, so we can find the index of one
such element using a variant of Grover search~\cite{bbht:tight-qsearch} in
$O(\sqrt{n/b}\log S')$ time since there will be at least $b$ potential indices.
This requires only $O(\log n)$ additional space over the cost of storing the
binary search tree.
We record that this boundary element occurs along the common elements in the
binary search tree, remove it from the set of allowable answers and repeat.   
This takes $O(\sqrt{n/(b-1)}\log S')$ time.
We continue $b$ times until we have found all $b$ of the boundary elements
that appear in the common vector.
The total cost of these Grover searches is $O(\sqrt{nb}\log S')$ which is
$O(\sqrt{nS'}\log S')$.
Since this is being done $n/S'$ times, the total time for this part of the
algorithm is $O(n^{3/2}(\log S')/\sqrt{S'})$, which is
$O(n^{3/2}\log^{3/2} n/S^{1/2})$.   
This yields the claimed time bound.
\end{proof}

The above  algorithm and the $TS =\Omega(n^2)$ classical time-space tradeoff
lower bound for $(F_0 \bmod 2)^{\boxplus n}$ prove that quantum computers have
an advantage over classical ones in sliding-windows.

\section{Order Statistics in Sliding Windows}\label{sec:order}

We first show that when order statistics are extreme, their complexity over
sliding windows does not significantly increase over that of a single
instance.

\begin{theorem}
\label{maxUBthm}
There is a deterministic comparison algorithm that computes
$MAX_n^{\boxplus n}$ (equivalently $MIN_n^{\boxplus n}$) using
time $T\in O(n\log n)$ and space $S\in O(\log n)$.
\end{theorem}

\begin{proof}
Given an input $x$ of length $2n-1$, we consider the window of $n$
elements starting at position $ \lceil\frac{n}{2} \rceil$ and ending
at position $n+\lceil\frac{n}{2}\rceil-1$ and find the largest element
in this window naively in time $n$ and space $O(\log n)$; call it $m$.
Assume without loss of generality that $m$ occurs between positions
$ \lceil\frac{n}{2} \rceil$ and $n$, that is, the left half of the
window we just considered. Now we slide the window of length $n$ to
the left one position at a time. At each turn we just need to look
at the new symbol that is added to the window and compare it to $m$.
If it is larger than $m$ then set this as the new maximum for that
window and continue.

We now have all outputs for all windows that start in positions 1 to
$\lceil\frac{n}{2} \rceil$. For the remaining outputs, we now run
our algorithm recursively on the remaining
$n+\lceil\frac{n}{2}\rceil$-long region of the input.  We only need to maintain the left and right endpoints of the current region.
At each level in the recursion, the number of outputs is halved and
each level takes $O(n)$ time. Hence, the overall time complexity is
$O(n\log n)$ and the space is $O(\log n)$.
\end{proof}

In contrast when an order statistic is near the middle, such as the median,
we can derive a significant separation in complexity between that of the
sliding-window version and that of a single instance.
This follows by a simple reduction and known
time-space tradeoff lower bounds for sorting~\cite{bc:sorting,bea:sorting}.

\begin{theorem}
Let $B$ be a branching program computing $O_t^{\boxplus n}$ in time
$T$ and space $S$ on an input of size $2n-1$, for any $t \in [n]$.
Then $T\cdot S \in \Omega(t^2)$ and the same bound applies to expected time
for randomized algorithms.
\end{theorem}

\begin{proof}
We give lower bound for $O_t^{\boxplus n}$ for $t \in[n]$ by showing a
reduction from sorting.
Given a sequence $s$ of $t$ elements to sort taking values in
$\{2, \ldots, n-1\}$, we create a $2n-1$ length string as follows: the first
$n-t$ symbols take the same value of $n$, the last $n-1$ symbols take
the same value of 1 and we embed the $t$ elements to sort in the
remaining $t$ positions, in an arbitrary order.
For the first window, $O_t$ is the maximum of the sequence $s$. As
we slide the window, we replace a symbol from the left, which has
value $n$, by a symbol from the right, which has value 1. The
$t^{th}$ smallest element of window $i = 1, \ldots,t $ is the
$i^{th}$ largest element in the sequence $s$. Then the first $t$
outputs of $O_t^{\boxplus n}$ are the $t$ elements of the sequence
$s$ output in increasing order.
The lower bound follows from~\cite{bc:sorting,bea:sorting}.
As with the bounds in~Corollary~\ref{average-and-random}, the proof methods
in~\cite{bc:sorting,bea:sorting} also immediately extend to average case and
randomized complexity.
\end{proof}

%

For the median ($t = \lceil n/2\rceil$), 
there is an errorless randomized algorithm for the single-input version 
with $\overline T \in O(n \log \log_S n)$ for $S\in \omega(\log n)$
and this is tight for comparison algorithms~\cite{chan:selection-journal}.

\section{Discussion}

We have shown a new sharper $T\in\tilde O(n^{3/2}/S^{1/2})$ upper bound for
randomized branching programs solving the element distinctness problem. 
Our algorithm is also implementable in similar time and space by RAM algorithms
using input randomness.   The impediment in implementing it on a RAM
with online randomness is our use of truly random hash functions $h$.

It seems plausible that a similar analysis would hold if those hash functions
were replaced by some $\log^{O(1)} n$-wise independent hash function such
as $h(x)=(p(x)\bmod m)\bmod n$ where $p$ is a random polynomial of degree
$\log^{O(1)}n$, which can be specified in space $\log^{O(1)}n$ and evaluated
in time $\log^{O(1)}n$, which would suffice for a similar
$T\in\tilde O(n^{3/2}/S^{1/2})$ randomized RAM algorithm with the most
natural online randomness.    The difficulty in analyzing this is the
interaction of the chaining of the hash function $h$ with the values of the
$x_i$.

It remains to be able to produce a time-space tradeoff separation for the
single-output rather than between the windowed versions of $ED$ and either
$F_k$ or $F_0\bmod 2$.  
In the very different context of quantum query complexity
two of the authors of this paper showed another separation between the
complexities of the $ED$ and $F_0\bmod 2$ problems.
$ED$ has quantum query complexity $\Theta(n^{2/3})$ (lower bound
in~\cite{as:quantum-collision} and matching quantum query algorithm
in~\cite{ambainis:distinctness}).  On other hand, the 
lower bounds in~\cite{bm:qac0} imply that $F_0\bmod 2$ has quantum query
complexity $\Omega(n)$.

\section*{Acknowledgements}
The authors would like to thank Aram Harrow, Ely Porat and  Shachar Lovett for a number of insightful discussions and helpful comments during the preparation of this paper.
\bibliographystyle{initials-plain}
\bibliography{theory,extra}

\appendix
\newpage
\section{A fast average case algorithm for $ED^{\boxplus n}$ with
alphabet $[n]$}
\label{sec:average}

We now give the algorithm and proof for Theorem~\ref{average-window-ED}.
The simple method we employ is as follows.  We start at the first window of
length
$n$ of the input and perform a search for the first duplicate pair starting
at the right-hand end of the window and going to the left.
We check if a symbol at position $j$ is involved in a duplicate by simply

scanning all the symbols to the right of position $j$ within the window.
If the algorithm finds a duplicate in a suffix of length $x$, it shifts the
window to the right by $n-x+1$ and repeats the procedure from this point.
If it does not find a duplicate at all in the whole window, it simply moves
the window on by one and starts again.

In order to establish the running time of this simple method,
we will make use of the following birthday-problem-related facts.

\begin{lemma}\label{lemma:birthday}
Assume that we sample i.u.d.\@ with replacement from $[n]$ with $n\ge 4$.
Let $X$ be a discrete random variable
that represents the number of samples taken when the first duplicate is found.
Then
\begin{equation}
\Pr\left(X \geq n/2\right) \leq e^{-\frac{n}{16}}. \label{birthday:0}
\end{equation}
and
\begin{equation}
\mathbb{E}(X^2) \leq 4n. \label{birthday:2}
\end{equation}
\end{lemma}
\begin{proof}
For every $x\ge 1$ we have 
\begin{equation*}
\Pr(X \geq x) =    \prod_{i=1}^{x-1}\left(1-\frac{i}{n}\right)
        \leq   \prod_{i=1}^{x-1} e^{-\frac{i}{n}}
        \leq e^{-\frac{x^2}{4n}}.
\end{equation*}
Inequality~\eqref{birthday:0} now follows by substituting $x= n/2$ 
giving
\[
 \Pr\left(X \geq \frac{n}{2}\right) \leq e^{-\frac{n}{16}}.
\]
To prove inequality~\eqref{birthday:2}, recall that for non-negative valued
discrete random variables
\[
\mathbb{E}(X) = \sum_{x=1}^{\infty} \Pr(X \geq x).
\]
Observe that
\begin{align*}
 \mathbb{E}(X^2) &= \sum_{x=1}^{\infty} \Pr(X^2 \geq x)
    = \sum_{x=1}^{\infty} \Pr(X \geq \sqrt{x})\\
    &\leq \sum_{x=1}^{\infty} e^{-\frac{(\sqrt{x})^2}{4n}}
    \leq \int_{x=0}^{\infty} e^{-\frac{(\sqrt{x})^2}{4n}}
    = 4n.
\end{align*}
\end{proof}

We can now show the running time of our average case algorithm for
$ED^{\boxplus n}$.

\begin{theorem}
For input sampled  i.u.d.\@ with replacement from alphabet $[n]$,
$ED^{\boxplus n}$ can be solved in average time $\overline{T} \in O(n)$ 
and average space $\overline{S} \in O(\log{n})$.
\end{theorem}

\begin{proof}
Let $U$ be a sequence of values sampled uniformly from $[n]$ with $n\geq 4$.
Let $M$ be the index of the first duplicate in $U$ found when scanning from
the right and let $X=n-M$.
Let $W(X)$ be the number of comparisons required to find $X$.
Using our naive duplicate finding method we have that $W(X) \leq X(X+1)/2$.
It also follows from inequality~\eqref{birthday:2} that $\mathbb{E}(W) \leq 4n$.

Let $R(n)$ be the total running time of our algorithm and note that
$R(n) \leq n^3/2$.
Furthermore the residual running time at any intermediate stage of the
algorithm is at most $R(n)$.

Let us consider the first window and let $M_1$ be the index of the first
duplicate from the right and let $X_1 = n - M_1$.
If $X_1 \geq n/2$, denote the residual running time by $R^{(1)}$.
We know from~\eqref{birthday:0} that
$\Pr(X_1 \geq n/2) \leq e^{-\frac{n}{16}}$.
If $X_1 < n/2$, shift the window to the right by $M_1+1$ and find $X_2$ for
this new window.
If $X_2 \geq n/2$, denote the residual running time by $R^{(2)}$.
We know that $\Pr(X_2 \geq n/2) \leq e^{-\frac{n}{16}}$.
If $X_1 < n/2 $ and $X_2 < n/2$ then the algorithm will terminate,
outputting `not all distinct' for every window.

The expected running time is then
\begin{equation*}\begin{split}
\mathbb{E}(R(n)) &= E\left(W(X_1)\right) + E\left(R^{(1)}\right)\Pr\left(X_1 \geq \frac{n}{2}\right)\\
&\quad+\Pr\left(X_1 < \frac{n}{2}\right)\Bigl[E\left(W(X_2) \middle\vert X_1 < \frac{n}{2}\right) + E\left(R^{(2)}\right) \Pr\left(X_2 \geq \frac{n}{2} \middle\vert X_1 < \frac{n}{2}\right)\Bigr]\\
&\leq 4n + \frac{n^3}{2} e^{-\frac{n}{16}} + 4n + \frac{n^3}{2} e^{-\frac{n}{16}} \in O(n)
\end{split}\end{equation*}

The inequality follows from the followings three observations.
We know trivially that $\Pr(X_1 < n/2) \leq 1$.
Second, the number of comparisons $W(X_2)$ does not increase if some of
the elements in a window are known to be unique. Third,
$\Pr(X_2 \geq n/2 \land X_1 < n/2) \leq \Pr(X_2 \geq n/2) \leq e^{-\frac{n}{16}}$.
\end{proof}

We note that similar results can be shown for inputs uniformly chosen from the
alphabet $[cn]$ for any constant $c$.

\end{document}